\documentclass[format=acmsmall, screen=true, nonacm=true]{acmart}
\usepackage{geometry}
\usepackage{subfig}
\usepackage{graphicx}

\renewcommand{\arraystretch}{1.5}
\usepackage{graphicx} %
\usepackage{amsmath}
\usepackage{amsthm}
\usepackage{url}
\usepackage{stmaryrd,hyperref}
\usepackage{color}
\usepackage{xspace}
\usepackage{thmtools}
\usepackage{tikz-cd}
\usepackage{mathtools}
\usepackage{booktabs}
\usepackage{wasysym}
\usepackage{bm}

\DeclareFontFamily{U}{mathb}{\hyphenchar\font45}
\DeclareFontShape{U}{mathb}{m}{n}{
      <5> <6> <7> <8> <9> <10> gen * mathb
      <10.95> mathb10 <12> <14.4> <17.28> <20.74> <24.88> mathb12
      }{}
\DeclareSymbolFont{mathb}{U}{mathb}{m}{n}
\DeclareMathSymbol{\precneq}{3}{mathb}{"AC}

\newtheorem{theorem}{Theorem}[section]
\newtheorem{fact}[theorem]{Fact}
\newtheorem{definition}[theorem]{Definition}
\newtheorem{example}[theorem]{Example}
\newtheorem{question}[theorem]{Question}
\newtheorem{remark}[theorem]{Remark}
\newtheorem{lemma}[theorem]{Lemma}
\newtheorem{proposition}[theorem]{Proposition}
\newtheorem{corollary}[theorem]{Corollary}

\newcommand{\sem}[1]{[\![#1]\!]}
\newcommand{\maj}{\mathsf{maj}}
\newcommand{\PROP}{\textsc{prop}}
\newcommand{\finops}[1]{\ensuremath{\textsc{FinOps}_{#1}}}
\newcommand{\class}[1]{\ensuremath{\mathsf{#1}}}
\newcommand{\logic}[1]{\ensuremath{\mathbf{#1}}}
\newcommand{\cloneof}[1]{[#1]}
\newcommand{\PL}{\ensuremath{\textup{\rm PL}}}
\newcommand{\ML}{\ensuremath{\textup{\rm ML}}}
\newcommand{\subf}{\ensuremath{\textsc{subf}}}
\newcommand{\tree}{\ensuremath{\textup{\rm tree}}}
\newcommand{\dagsize}{\ensuremath{\textup{\rm dag}}}
\newcommand{\MOD}{\ensuremath{\textup{\rm MOD}}}
\newcommand{\lab}{\ensuremath{\textup{\rm lab}}}
\newcommand{\oper}[1]{\circ_{#1}}
\newcommand{\expand}[1]{\mathsf{expand}(#1)}
\newcommand{\sharpp}{\#\mathrm{P}\xspace}
\newcommand{\equivclass}[2]{\textsf{equiv}_{#2}(#1)}

\usepackage[table]{xcolor}
\definecolor{hellgrau}{gray}{0.8}
\usepackage{xspace}
\newcommand{\eqdef}{\coloneqq}
\newcommand{\N}{\protect\ensuremath{\mathbb{N}}\xspace}

\newcommand{\true}{\protect\ensuremath{\top}\xspace}
\newcommand{\false}{\protect\ensuremath{\bot}\xspace}
\DeclareMathOperator{\nimp}{\nrightarrow}

\DeclareMathOperator{\xor}{\oplus}
\DeclareMathOperator{\eq}{\leftrightarrow}

\DeclareMathOperator{\limp}{\rightarrow}

\newcommand{\set}[1]{\ensuremath{\left\{#1\right\}}}
\newcommand{\cloneFont}[1]{\mathsf{#1}}
\newcommand{\CloneBF}{\protect\ensuremath{\cloneFont{BF}}}
\newcommand{\CloneM}{\protect\ensuremath{\cloneFont{M}}}
\newcommand{\CloneL}{\protect\ensuremath{\cloneFont{L}}}
\newcommand{\CloneR}{\protect\ensuremath{\cloneFont{R}}}
\newcommand{\CloneD}{\protect\ensuremath{\cloneFont{D}}}
\newcommand{\CloneN}{\protect\ensuremath{\cloneFont{N}}}
\newcommand{\CloneS}{\protect\ensuremath{\cloneFont{S}}}
\newcommand{\CloneV}{\protect\ensuremath{\cloneFont{V}}}
\newcommand{\CloneE}{\protect\ensuremath{\cloneFont{E}}}
\newcommand{\CloneI}{\protect\ensuremath{\cloneFont{I}}}

\newcommand{\aimp}{\mathsf{aimp}}
\newcommand{\oxor}{\mathsf{oxor}}
\DeclareMathOperator{\threeXor}{\xor^3}

\title{Modal Fragments}

\author{Nick Bezhanishvili}
\orcid{0009-0005-6692-5051}
\affiliation{
\institution{ILLC, University of Amsterdam}
\country{The Netherlands}}

\author{Balder ten Cate}
\orcid{0000-0002-2538-5846}
\affiliation{
\institution{ILLC, University of Amsterdam}
\country{The Netherlands}}

\author{Arunavo Ganguly}
\orcid{0009-0004-9098-9904}
\affiliation{
\institution{Department of Computing Science, Umeå Universitet}
\country{Sweden}}

\author{Arne Meier}
\orcid{0000-0002-8061-5376}
\affiliation{\institution{Institute of Theoretical Computer Science, Leibniz Universität Hannover}\country{Germany}}

\begin{document}

\begin{abstract}
We survey systematic approaches to basis-restricted fragments of propositional logic and modal logics, with an emphasis on how expressive power and computational complexity depend on the allowed operators. The propositional case is well-established and serves as a conceptual template: Post's lattice organizes fragments via Boolean clones and supports complexity classifications for standard reasoning tasks. For modal fragments, we then bring together two historically independent lines of investigation: a general framework where modal fragments are parameterized by a basis of  ``connectives'' defined by arbitrary modal formulas (initially proposed and studied by logicians such as Kuznetsov and Ra\c{t}\v{a} in the 1970s), and the more tractable class of what we call \emph{simple} modal fragments parameterized by Boolean functions plus selected modal operators, where Post-lattice methods enable systematic decidability and dichotomy results. Along the way, we collect and extend results on teachability and exact learnability from examples for both propositional fragments and simple modal fragments, and we conclude by identifying several open problems.
\end{abstract}

\maketitle

\section{Introduction}

There is a long tradition, going back to early work by Emil Post~\cite{Post}, Jablonski~et~al.~\cite{Jablonski1970}, and Harry Lewis~\cite{Lewis_Sat_Problems_for_propositional_calculi}, of studying fragments of propositional logic parameterized by a basis of Boolean functions. In this line of research, one investigates how expressive power and the computational complexity of fundamental algorithmic tasks depend on the chosen basis. Post’s lattice provided a systematic classification of Boolean function bases and became a foundational tool for understanding such dependencies.

This perspective was subsequently extended to richer logical systems, including modal logics. In fact, two largely independent lines of work emerged investigating fragments of modal logic in a systematic way.
The first line was initiated by Alexander Kuznetsov~\cite{kuznetsov1971functional} and Metodie Ra{\c{t}}\v{a}~\cite{Ratsa71} in the 1970s. In this line, a fragment of modal logic is given by a basis of ``connectives'' specified by arbitrary modal formulas.
An example of such a fragment is the ``contingency fragment'' 
\[\phi \Coloneqq x \mid \neg \phi\mid \phi\land\psi\mid \CIRCLE\phi,\] where $\CIRCLE\phi$ is defined by $\Diamond \phi\land\Diamond\neg\phi$.
One can equivalently think of such a fragment as obtained by taking a finite set of formulas (in this case: $x, \neg x$, $x\land y$ and $\Diamond x\land \Diamond\neg x$) and closing off under uniform substitution.
Kuznetsov and Ra\c{t}\v{a} (and other subsequent researchers) 
were interested in the structure of the lattice of all such fragments, as well as in the decidability of
 meta-problems such as \emph{expressive completeness} (given a fragment, is it expressively complete?) and 
\emph{expressive containment}
(given two fragments, is one  contained in the other, modulo equivalence?).
Owing to the generality of this framework, however, many of the results obtained tend to be either negative (for instance showing undecidability for certain meta-problems), or restricted in scope (for instance, applying only to locally tabular logics such as \logic{S5}).

In the 1990s, Jeavons and collaborators~\cite{DBLP:conf/cocoon/JeavonsC95,DBLP:journals/amai/JeavonsCP98}, along with Nadia Creignou and colleagues~\cite{DBLP:journals/jcss/Creignou95,DBLP:journals/iandc/CreignouH96,DBLP:journals/ita/CreignouH97,DBLP:books/daglib/0004131}, revitalized the line of research initiated by Lewis~\cite{Lewis_Sat_Problems_for_propositional_calculi} and Schaefer~\cite{DBLP:conf/stoc/Schaefer78}, ultimately leading to a refined formulation of Schaefer’s theorem~\cite{DBLP:journals/jcss/AllenderBISV09}.
This inspired work on fragments of temporal~\cite{DBLP:conf/stacs/BaulandHSS06} and modal logics~\cite{DBLP:journals/eccc/BaulandSSSV06}, starting the aforementioned second, independent line of research on modal fragments. 
This line of research involves  a more restrictive yet more tractable framework. Here, modal fragments are parameterized by a set of \emph{Boolean functions} plus a (separately specified)  subset of the modal operators. 
We will call such modal fragments \emph{simple} in this paper.
The aforementioned contingency fragment is not of this form, but, for example, 
the ``monotone fragment''
\[\phi \Coloneqq x \mid \top\mid\bot\mid \phi\land\psi \mid \phi\lor\psi\mid \Diamond\phi\mid\Box\phi\] is.
This second line is technically much closer to Post’s original setting and has proved significantly more amenable to algorithmic analysis, aided by Post's lattice. Consequently, within this framework, numerous dichotomy theorems have been obtained, yielding complete complexity classifications for  decision problems associated with modal fragments.
More recently, this second line of research has been extended beyond basic modal logic to other logical formalisms, including temporal logics, hybrid logics, and team logics, further demonstrating the robustness of the Boolean-basis perspective.

The purpose of this article is threefold: (i) to bring these two lines of work together into a unified conceptual narrative, (ii) to survey key results from both traditions, and (iii) to outline directions for future research. We also present some new results pertaining to the learnability and teachability of formulas by examples for different fragments of propositional and modal logic.

\paragraph{Outline}
In Section~\ref{sec:prel}, we review relevant background on complexity classes and on Boolean clones. In Section~\ref{section:Prop}, we review results on fragments of propositional logic, which form a backdrop for the next part of the paper.
In Section~\ref{section:Modal}, we consider fragments of modal logic and we review results from both traditions mentioned above.
In Section~\ref{sec:other_logics}, we briefly discuss a number of other logical systems for which fragments have been similarly studied. Finally, we conclude in Section~\ref{sec:conclusion} by listing a number of open problems.

\section{Preliminaries}
\label{sec:prel}

\subsection{Relevant complexity classes}
In this paper, we will be surveying dichotomy results regarding the complexity of various algorithmic problems for fragments of propositional logic and of modal logic. These will 
make use of a number of standard complexity classes, which we briefly review here.

The complexity class $\class{P}$ is the class consisting of problems solvable in \emph{deterministic polynomial time}; $\class{NP}$ and $\class{coNP}$ are the class of problems solvable in nondeterministic polynomial time, and its complement; $\class{PSPACE}$ is the class solvable in polynomial space (deterministically, or, by Savitch's theorem, equivalently, nondeterministically polynomial space); $\class{EXPTIME}$ is the class solvable  deterministically in $2^{n^{O(1)}}$ time; the class $\Theta_2^{\class{P}}=\class{P}^{\class{NP}[\log]}$ is defined via logarithmic many oracle calls to an $\class{NP}$-language, and it has been shown that this class coincides with $\class{P}^{||\class{NP}}$, which is unrestricted but non-adaptive (i.e., parallel) queries~\cite{DBLP:journals/iandc/BussH91,DBLP:journals/jcss/Hemachandra89}.
The counting class $\sharpp$ is the class of problems associated with decision problems in $\class{NP}$. 
More precisely, it is the class that asks to compute the number of accepting paths of a nondeterministic Turing machine on a given input.

Below $\class{P}$, we have the classes  $\class{NL}$  for algorithms solvable in nondeterministic logarithmic space; and
$\class{L}$  consisting of problems solvable in  deterministic logarithmic space; 
In-between $\class{L}$ and $\class{P}$, we also find
$\class{\oplus L}$ (Parity-L) consisting of the problems for which there is a nondeterministic log-space Turing machine such that  the machine has an odd number of accepting runs on a given input if and only if the input is a yes-instance.

Below $\class{L}$ are various \emph{circuit classes}, 
of which the following will be relevant for us:
$\class{AC}^0$ is the class of languages decidable by families of constant-depth, polynomial-size Boolean circuits with $\wedge$-gates, $\vee$-gates and $\neg$-gates, where the fan-in of the gates of the first two types is unbounded; $\class{AC}^0[2]$ extends $\class{AC}^0$ with unbounded fan-in \emph{parity} (mod~2) gates; $\class{NC}^1$ consists of  polynomial-size, depth $O(\log n)$ circuits with bounded fan-in. For these circuit classes, we will assume logspace-uniformity unless stated otherwise. The above complexity classes then form the following hierarchy: 
\[
\class{AC}^0 \subsetneq \class{AC}^0[2] \subsetneq \class{NC}^1 
\subseteq \class{L} 
\begin{array}{c}
  \rotatebox[origin=B]{15}{$\subseteq$} \\[-6pt]
  \rotatebox[origin=B]{-15}{$\subseteq$}
\end{array}\hspace*{-1.5ex}
\begin{array}{c}
\class{NL}\\[-2pt]
\class{\oplus L}
\end{array}
\hspace*{-1.5ex}
\begin{array}{c}
  \rotatebox[origin=B]{-15}{$\subseteq$} \\[-6pt]
  \rotatebox[origin=B]{15}{$\subseteq$}
\end{array}
\class{P}
\begin{array}{c}
  \rotatebox[origin=B]{15}{$\subseteq$} \\[-6pt]
  \rotatebox[origin=B]{-15}{$\subseteq$}
\end{array}\hspace*{-1.5ex}
\begin{array}{c}
\class{NP}\\[-2pt]
\class{coNP}
\end{array}
\hspace*{-1.5ex}
\begin{array}{c}
  \rotatebox[origin=B]{-15}{$\subseteq$} \\[-6pt]
  \rotatebox[origin=B]{15}{$\subseteq$}
\end{array}
\Theta_2^{\class{P}}\subseteq\class{PSPACE}\subseteq \class{EXPTIME}.
\]

For each of these complexity classes, suitable notions of ``hardness'' and ``completeness'' can be defined, based on various types of
reductions. For the complexity classes
below $\class{L}$, there is not one canonical type of reduction, and different types of reductions are used in different papers. 
In this context, sometimes more fine-grained notions of reducibility have to be used, such as the following. A language $L$ is \emph{$\class{AC^0}$ many-one
reducible} to a language $B$ if there exists a function $f$ computable by a log time-uniform $\class{AC^0}$-circuit family such that  $x\in L$ if and only if $f(x)\in B$.  
When reviewing such hardness results from the literature, we will not always be precise about the types of reduction used. Instead we will refer the reader to the respective paper for the full details.

For full definitions and background, see \cite{AroraBarak2009} or \cite{Vollmer1999-VOLITC-2}.
\subsection{Clones}

Informally, a \emph{clone} on a set 
$A$ is a collection of finitary operations on 
$A$ that contains all projection functions and is closed under arbitrary composition. Formally,
a \textit{finitary operation} on a set $A$ any function $f\colon A^n\to A$ with $n>0$. A \emph{projection} 
is a finitary operation  $f\colon A^n\to A$ defined by
$f(x_1, \ldots, x_n)=x_k$ for some fixed $k\leq n$.
Such a function is 
usually denoted by 
$\pi^n_k$. We denote the set of all such finitary operations 
on $A$ by $\finops{A}$.

\begin{definition}[Clones]
 A \emph{clone} over a set $A$ is a subset $C\subseteq \finops{A}$ that satisfies:
\begin{enumerate}
    \item $C$ contains all projections $\pi^n_k$ (for all $k$ and $n$ with $k\leq n$), and
    \item $C$ is closed under composition: for all $g_1,\dots g_m, f\in C$, where $f$ is an $m$-ary operation and each $g_i$ is an $n$-ary operation, 
    there is an $n$-ary operation $h\in C$ such that 
     $h(x_1, \ldots, x_n) = f(g_1(x_1,\ldots x_n), \ldots ,g_m(x_1,\ldots, x_n))$.
\end{enumerate}
\end{definition}

The above properties are preserved by  set-theoretic intersection.
\begin{fact}\label{fact:Clones}
    Let $A$ be any set.
    \begin{enumerate}
        \item The entire set $\finops{A}$ is a clone over $A$;
        \item For every collection of clones $\{\,C_i\mid i\in I\,\}$, the intersection
    $C=\bigcap_{i\in I}C_i$  is again a clone.
    \end{enumerate}
\end{fact}

Fact \ref{fact:Clones} provides the notion of a \textit{clone generated by a set of finitary operations}.  Formally, for any $O\subseteq \finops{A}$, we define
$\cloneof{O}\coloneqq\bigcap\{\,C \mid \text{$O\subseteq C$ and $C$ is a clone}\,\}$. 

Although the set-theoretic union of clones is not necessarily a clone, 
set-theoretic union can be coupled with the closure operation $\cloneof{\cdot}$ to define a join operation on clones: 
let $\bigsqcup_{i\in I} C_i = \cloneof{\bigcup_{i\in I} C_i}$. 
Since the $\cloneof{\cdot}$ operation distributes over set-theoretic intersection $\cap$, the meet of clones is defined as: $\bigsqcap_{i\in I}C_i = \cloneof{\bigcap_{i\in I}C_i}= \bigcap_{i\in I}C_i.$ 
The operations $\sqcap$ and $\sqcup$ provide a lattice structure to the set of all clones.
\begin{fact}\label{fact:clones-lattice}
    Let $\mathcal{C}_A$ denote the set of all clones over the set $A$. Then $(\mathcal{C}_A, \sqcup, \sqcap)$ forms a complete lattice. The partial order induced by  this lattice is the set-theoretic inclusion ($\subseteq$).   
\end{fact}

Recall that, in universal algebra~\cite{Universal_Algebra}, an \emph{algebra} is  
a tuple
$\mathcal{A}=(A,f_1, \ldots, f_n)$ where $A$ is a set and 
$f_i\colon A^{k_i}\to A$, where $k_i\geq 0$ is the arity of $f_i$. By 
\emph{the clone generated by $\mathcal{A}$}, 
also denoted $\cloneof{\mathcal{A}}$, we will mean the clone generated by the functions of $\mathcal{A}$, i.e., $\cloneof{\mathcal{A}}=\cloneof{\{f_1, \ldots, f_n\}}$. This definition assumes that all functions $f_i$ have positive arity, so that they qualify as finitary operations. In case the algebra includes zero-ary functions (i.e., constants), these are treated as constant unary functions (i.e., unary functions that return the same value on all inputs). 

\begin{example}
\label{example:algebras-clones}
Consider the algebra $\mathcal{A}=(\mathbb{N},+,0,1)$. The clone generated
by $\mathcal{A}$ is, by definition,
the clone over $\mathbb{N}$ generated by the addition function and the constant functions $c_0,c_1\colon \mathbb{N}\to\mathbb{N}$. It is, in fact,
simply the set of all linear functions with integer coefficients.
Similarly, consider the two-element Boolean algebra $\mathcal{B}=(\{0,1\},\land,\lor,\neg,0,1)$. The
clone generated by $\mathcal{B}$ is simply the entire set $\finops{\{0,1\}}$.
\end{example}

\begin{remark}\label{rem:zeroary}
The above definition of clones follows the standard treatment in the literature by disregarding 
zero-ary operations (or, rather, treating them as constant unary operations). 
With some adjustments, the framework of clone theory can be extended such that  zero-ary operations are included as first-class citizens. This perspective is more consistent with the way algebras are defined in universal algebra, where zero-ary terms (also known as constants) are  permitted.
In this extended view, every clone that excludes zero-ary operations can be regarded as a subclone of one that does allow them. 
Thus, for instance, Post's lattice (as defined in the next subsection, cf.~Figure~\ref{fig:Post_Lat}) would then become slightly larger.
Most results that we will discuss in this paper can be lifted in a natural way to this extended setting.
\end{remark}

The following notation is adopted to improve readability: For $O,O'\subseteq \finops{A}$, define 
$$O\preceq O' \iff \cloneof{O}\subseteq \cloneof{O'}.$$ We also write $O\equiv O'$ when $O\preceq O'$ and $O'\preceq O$.
We employ this notation throughout Section~\ref{section:Prop} and Section~\ref{section:Modal} to condense theorem statements. 
The relation $\preceq$ is a pre-order. The operation $\cloneof{\cdot}$ is monotone over $\finops{A}$, though not-necessarily strict; that is, if $O,O'\subseteq \finops{A}$ and $O\subseteq O'$ then $\cloneof{O}\subseteq \cloneof{O'}$, however $O\subsetneq O'$ need not imply $\cloneof{O}\subsetneq \cloneof{O'}$.

For a general introduction to clone theory, we refer to the textbook of Pippenger~\cite[Sect.~1.4]{DBLP:books/daglib/0092426}.
\subsection{Boolean Clones}
\label{sec:post}

 A \emph{Boolean clone} is a clone over the set $\{0,1\}$. In light of Example~\ref{example:algebras-clones}, we can equivalently define a Boolean clone to be any subclone of the clone generated by the two-element Boolean algebra.
 The lattice of all Boolean clones, i.e., $\mathcal{C}_{\{0,1\}}$ 
 is known as \emph{Post's lattice}, named after Emil Leon Post, whose seminal work \cite{Post} completely described its structure. 
  Post's lattice is depicted in Figure~\ref{fig:Post_Lat}. 
We list a few key facts about it, which were already established by Post in his early work, that will be important later on.

 \begin{fact}[Boolean clones are finitely generated] 
 \label{fact:finitely-generated}
 For every $C\in \mathcal{C}_{\{0,1\}}$ there is finite set $O\subseteq\finops{\{0,1\}}$ such that  $C=\cloneof{O}$.  
  In particular,  there are only countably many Boolean clones.
\end{fact}

\begin{fact}[All downward closed sets of Boolean clones are finitely generated] 
\label{fact:finite-downward-closure}
    Every downward-closed set of Boolean clones is the downward closure of a finite set of Boolean clones.
\end{fact}

Here, by a \emph{downward closed} set of Boolean clones we mean a set $\mathcal{F}$ such that  for all clones $C\subseteq C'$, $C'\in \mathcal{F}$ implies $C\in \mathcal{F}$; and the downward closure of a set $\mathcal{F}$ is the
smallest downward closed set $\mathcal{F}'$ such that  $\mathcal{F}\subseteq \mathcal{F}'$. For example, the set of all  proper subclones of a given Boolean clone $\mathcal{F}$ is a downward closed set, and hence can only have finitely many maximal elements. Fact~\ref{fact:finite-downward-closure} is a consequence of the fact that
the Boolean clones are well-quasi-ordered by $\supseteq$.

\begin{fact}[The Uniform Clone Membership Problem is Decidable]
\label{fact:clone-uniform-membership}
  The following two problems are decidable:
  \begin{itemize}
\item   Given a finite set $O\subseteq \finops{\{0,1\}}$ and $f\in\finops{\{0,1\}}$, is it the case that $f\in \cloneof{O}$?
\item   Given finite sets $O,O'\subseteq \finops{\{0,1\}}$, is it the case that $O\preceq O'$?
  \end{itemize}
\end{fact}

Indeed, the above two problems are easily seen to be computationally equivalent. Their exact complexity depends on the representation used for the finitary operations: they are solvable in polynomial time (and, in fact, belongs to the complexity class $\class{NC^1}$)
if they are specified by truth tables, while they are $\class{coNP}$-complete (for fixed sets $O$) or $\Theta^P_2$-complete (for $O$ as part of the input) if they are specified by Boolean circuits or by Boolean formulas~\cite{Vollmer_Complexity,DBLP:journals/mst/Jerabek21}.

Facts~\ref{fact:finite-downward-closure},~\ref{fact:finitely-generated}, and~\ref{fact:clone-uniform-membership} together imply:

\begin{fact}[All downward closed sets of Boolean clones are decidable] 
\label{fact:downward-closed-decidable}
    Let $\mathcal{F}$ be any downward-closed set of Boolean clones. Then it is decidable, given a finite set of Boolean functions $O$, whether $\cloneof{O}\in\mathcal{F}$. The same holds for upward closed sets of Boolean clones.
\end{fact}

The above facts show that the Post lattice is well-behaved.
In contrast, as we will
discuss in Section~\ref{sec:many-valued}, 
for sets $A$ with $|A|>2$, the lattice $\mathcal{C}_A$
no longer has these nice properties.

\begin{figure}
    \centering  \includegraphics[width=0.7\linewidth]{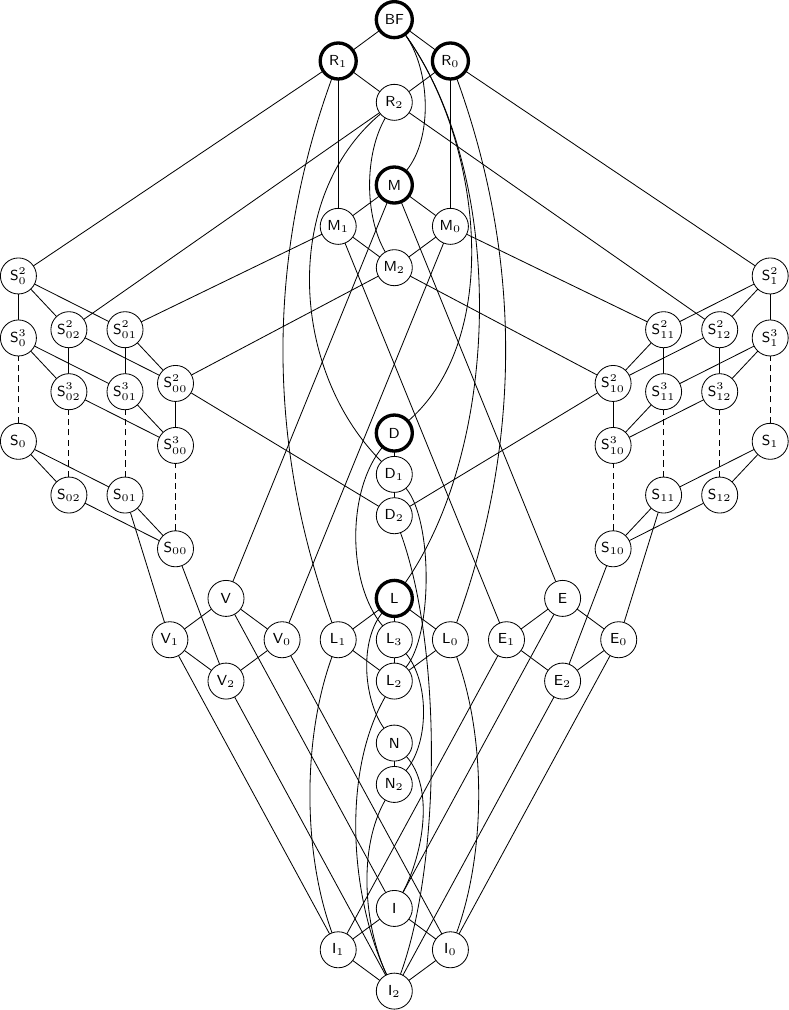}
    \caption{Post's Lattice}
    \label{fig:Post_Lat}
\end{figure}

\begin{table}
\rowcolors[]{2}{}{hellgrau}
\begin{center}  
      \resizebox{.7\linewidth}{!}{\renewcommand{\arraystretch}{1}$\begin{array}{clll}\toprule        \text{Class} & \text{Definition} & \text{Base(s)} & \\
        \midrule
        \CloneBF & \text{All Boolean functions} & \set{x\land y, \neg x} &\\
        \CloneR_0 & \set{f \mid f \text{ is \false-reproducing}} & \set{x\land y, x\xor y} &\\
        \CloneR_1 & \set{f \mid f \text{ is \true-reproducing}} & \set{x\lor y, x\eq y} &\\
        \CloneR_2 & \CloneR_0\cap\CloneR_1 & \set{x\lor y, x\land(y\eq z)} &\\
        \CloneM & \set{f \mid f \text{ is monotone}} & \set{x\lor y, x\land y, \false, \true} &\\
        \CloneM_0 & \CloneM\cap\CloneR_0 & \set{x\lor y, x\land y, \false} &\\
        \CloneM_1 & \CloneM\cap\CloneR_1 & \set{x\lor y, x\land y, \true} &\\
        \CloneM_2 & \CloneM\cap\CloneR_2 & \set{x\lor y, x\land y} &\\
        \CloneS_0 & \set{f \mid f \text{ is \false-separating}} & \set{x\limp y} &\\
        \CloneS_1 & \set{f \mid f \text{ is \true-separating}} & \set{x\nimp y} &\\
        \CloneS_0^n & \set{f \mid f \text{ is \false-separating of degree } n} & \set{x\limp y,T^{n+1}_2} &\\
        \CloneS_1^n & \set{f \mid f \text{ is \true-separating of degree } n} & \set{x\nimp y,T^{n+1}_n } &\\
        \CloneS_{00} & \CloneS_0\cap\CloneR_2\cap\CloneM & \set{x\lor(y\land z)}&\\\rowcolors[]{2}{}{hellgrau}
        \CloneS_{00}^n & \CloneS_0^n\cap\CloneR_2\cap\CloneM 
        &\set{x \lor (y \land z),{\mathrm{T}^{3}_2}}& \text{if } n=2,\\
        \cellcolor{white}&\cellcolor{white} &\cellcolor{white}\set{\mathrm{T}^{n+1}_2}&\cellcolor{white} \text{if }n\geq3\\
        \CloneS_{01} & \CloneS_0\cap\CloneM & \set{x\lor(y\land z),\true}&\\
        \CloneS_{01}^n & \CloneS_0^n\cap\CloneM & \set{T^{n+1}_2,\true}&\\
        \CloneS_{02} & \CloneS_0\cap\CloneR_2 & \set{x\lor(y\nimp z)}&\\
        \CloneS_{02}^n & \CloneS_0^n\cap\CloneR_2 & \set{x\lor(y\nimp z),T^{n+1}_2}&\\
        \CloneS_{10} & \CloneS_1\cap\CloneR_2\cap\CloneM & \set{x\land(y\lor z)}&\\\rowcolors[]{2}{}{hellgrau}
        \CloneS_{10}^n & \CloneS_1^n\cap\CloneR_2\cap\CloneM 
        &\set{x \land (y \lor z),\mathrm{T}^{3}_2}& \text{if }n=2,\\
        \cellcolor{white}&\cellcolor{white}&\cellcolor{white}\set{\mathrm{T}^{n+1}_n}&\cellcolor{white}\text{if }n\geq3\\
        \CloneS_{11} & \CloneS_1 \cap \CloneM  & \set{x \land (y \lor z), \false } &\\
        \CloneS_{11}^n & \CloneS_1^n \cap \CloneM  & \set{T^{n+1}_n, \false } &\\
        \CloneS_{12} & \CloneS_1 \cap \CloneR_2  & \set{\aimp(x,y,z)} &\\
        \CloneS_{12}^n & \CloneS_1^n \cap \CloneR_2  & \set{\aimp(x,y,z), T^{n+1}_n } &\\
        \CloneD & \set{f \mid f \text{ is self-dual}} & \set{ \maj(x, \overline y, \overline z) },\set{\maj,\lnot} &\\
        \CloneD_1 & \CloneD \cap\CloneR_2 & \set{ \maj(x, y, \overline z) } &\\
        \CloneD_2 & \CloneD \cap\CloneM & \set{ \maj(x, y, z)  } &\\
        \CloneL & \set{f \mid f \text{ is linear}} & \set{x \xor y,\true} &\\
        \CloneL_0 & \CloneL \cap\CloneR_0 & \set{x\xor y} &\\
        \CloneL_1 & \CloneL \cap\CloneR_1 & \set{x\eq y} &\\
        \CloneL_2 & \CloneL \cap\CloneR_2 & \set{\threeXor(x,y,z)} &\\ 
        \CloneL_3 & \CloneL \cap\CloneD & \set{x\xor y\xor z\xor \true} &\\
        \CloneV & \set{f \mid f \text{ is a disjunction or constant}} & \set{x\lor y, \false,\true } &\\
        \CloneV_0 & \CloneM_0\cap\CloneV & \set{ x\lor y, \false } &\\
        \CloneV_1 & \CloneM_1\cap\CloneV & \set{ x\lor y, \true } &\\
        \CloneV_2 & \CloneM_2\cap\CloneV & \set{ x\lor y} &\\
        \CloneE & \set{f \mid f \text{ is a conjunction or constant}} & \set{x\land y, \false, \true } &\\
        \CloneE_0 & \CloneM_0\cap\CloneE & \set{ x\land y, \false } &\\
        \CloneE_1 & \CloneM_1\cap\CloneE & \set{ x\land y, \true } &\\
        \CloneE_2 & \CloneM_2\cap\CloneE & \set{ x\land y} &\\
        \CloneN & \set{f \mid f \text{ depends on at most one variable}} & \set{\neg x,\false,\true} &\\
        \CloneN_2 & \CloneL_3\cap\CloneN & \set{ \neg x} &\\
        \CloneI & \set{f \mid f \text{ is a projection or a constant}} & \set{\false,\true}&\\
        \CloneI_0 & \CloneR_0\cap\CloneI & \set{\false}&\\
        \CloneI_1 & \CloneR_1\cap\CloneI & \set{\true}&\\
        \CloneI_2 & \CloneR_2\cap\CloneI & \set{}&\\\bottomrule
      \end{array}$}
\end{center}
\caption{List of all Boolean clones with their bases. For $n\ge m\in\N$, the threshold function $T^{n}_m\eqdef \bigvee_{\substack{S\subseteq\{1,\dots,n\},\\|S|=m}}\bigwedge_{i\in S}x_i$ requires $m$ bits out of $n$ set to $\true$.}\label{tbl:posts_bases}
\end{table}

As we will see in Section~\ref{section:Prop}, there is a close correspondence between Boolean clones and syntactic fragments of propositional logic. We will make use of this correspondence extensively.

\section{Fragments of Classical Propositional Logic}\label{section:Prop}

In this section, we consider syntactic
fragments of classical propositional logic (henceforth simply: propositional logic) obtained by fixing a set of allowed  connectives. For the syntax of such fragments, we require a bit of notation. That is, with every Boolean function $f$, we associate a Boolean connective $\oper{f}$ in the obvious way (also see a bit below for the semantics). 
Now, for a set $O$ of Boolean functions, 
we define $\PL_O$ to be the propositional language where we
allow (only) the functions in $O$ as connectives. 
In other words, $\PL_O$ is propositional language whose 
formulas are generated by the grammar
\[ \phi \Coloneqq x \mid \oper{f}(\phi_1, \ldots, \phi_n)\]
where $x$ is a propositional variable (from some fixed countably infinite set), and $f\in O$
an $n$-ary Boolean function.
For a given a finite set  $\PROP$ of propositional variables, we will 
also denote by $\PL_O[\PROP]$ the set of all
$\PL_O$-formulas using only propositional variables from the set $\PROP$. 
Thus, each choice of the set $O$ of Boolean function gives rise to a different language $\PL_O$. 
The semantics for these languages is as expected: given a truth assignment $V$,
$\sem{x}_V=V(x)$ for all propositional variables $x$, and $\sem{\oper{f}(\phi_1, \ldots,\phi_n)}_V=f(\sem{\phi_1}_V, \ldots, \sem{\phi_n}_V)$. 
With a slight abuse of notation,
we will write $\sem{\phi}$ also 
for the Boolean function that maps every truth assignment $V$ to the truth value $\sem{\phi}_V$.
As usual, two formulas,
$\phi, \phi'$
are said to be \emph{equivalent} if $\sem{\phi}=\sem{\phi'}$.

We use certain abbreviations for Boolean functions to enhance readability, readers are referred to Table~\ref{table:List_of_Abbreviations} for the exhaustive list.   
Note that $\top$ and $\bot$ are treated here as (constant) unary operations (cf.~Remark~\ref{rem:zeroary}).

\begin{figure}
   \subfloat{\begin{tabular}{cc}
    \toprule
        \text{Abbreviation} & \text{Corresponding formula}\\
        \midrule
        $\maj$ & $(x\land y)\lor (y\land z)\lor(x\land z) $\\
        $\aimp$ & $x\land (y\rightarrow z)$\\
        $\oxor$ & $x\lor (y \oplus z)$\\
        $\oplus$ &  $(\neg x\land y)\lor (x\land \neg y )$\\
        $\threeXor$ & $x\oplus y \oplus z$\\
        $\neg$ & $\neg x$
        \\
    \bottomrule\\
    \end{tabular}
    }
    ~~~~~~~
    \centering
    \subfloat{\begin{tabular}{cc}
    \toprule
        \text{Abbreviation} & \text{Corresponding formula}\\
        \midrule
        $\land$ & $(x\land y)$\\
        $\lor$ & $(x\lor y)$\\
        $\leftrightarrow$ & $(x\leftrightarrow y)$\\
        $\rightarrow$ & $(x\rightarrow y)$\\
        $\nrightarrow$ & $(x\land\neg y)$\\
        $\perp$ & $(x\land \neg x)$\\
        $\top$ & $(x\lor \neg x)$\\
    \bottomrule
    \end{tabular}}
    \caption{List of Abbreviations}
    \label{table:List_of_Abbreviations}
\end{figure}
\begin{example}
    The formula $p\land q$ belongs to the propositional fragment $\PL_{\{\land \}}$. 
    It does not belong to the propositional
    fragment $\PL_{\{\lor\}}$, and (as can be easily shown) it is not equivalent to a formula in the latter
    fragment either.
\end{example}

The \emph{size} of a $\PL_O$-formula can be measured in two ways: using a DAG-style representation or using a tree-style representation. Formally:

\[\begin{array}{lll}
    |x|_{\tree} &\eqdef& 1 \\
    |f(\phi_1, \ldots, \phi_n)|_{\tree} &\eqdef& 1+\Sigma_i|\phi_i|_{\tree} 
\end{array}\]
whereas $|\phi|_{\dagsize}$ is the number of distinct subformulas of $\phi$, i.e., $|\phi|_{\dagsize}=|\subf(\phi)|$ where
\[
    \subf(\phi) = 
\begin{cases}
    \{x\} & \text{for $\phi$ of the form $x$} \\
    \{\phi\}\cup\bigcup_i \subf(\phi_i) & \text{for $\phi$ of the form $f(\phi_1, \ldots, \phi_n)$}
\end{cases}\]

Equivalently, $|\phi|_{\tree}$ is the number of occurrences of propositional variables and connectives in $\phi$.
In terms of Boolean circuits, 
$|\phi|_{\dagsize}$ corresponds to the
size of a Boolean circuit 
computing $\phi$, where the circuit 
contains one node per subformula of $\phi$, while $|\phi|_{\tree}$ corresponds to the size of a Boolean circuit where each node has fan-out 
at most 1.

In the following subsections, we will be surveying different properties of the languages $\PL_O$ (based on the set $O$): 
expressive power and succinctness, complexity of reasoning tasks, and teachability and learnability.

\subsection{Expressive power and succinctness}\label{section: Expressive}

Let us say that $\PL_O$ is \emph{expressively contained} in $\PL_{O'}$ if every formula of 
$\PL_O$ is equivalent to some formula of $\PL_{O'}$. We say that $\PL_O$ and $\PL_{O'}$ are \emph{expressively equivalent} if they
are expressively contained in each other.
Then, it is easy to see that:

\begin{fact}\label{fact:containment}
    For all sets $O, O'$ of Boolean functions,
    $\PL_O$ is expressively contained in $\PL_{O'}$ if and only if $O\preceq O'$. Moreover, for any given finite sets $O, O'$ it is decidable whether this holds.
\end{fact}

In other words, the expressive power of $\PL_O$ is determined by the clone generated by $O$. Indeed, this fact follows from the very definition of clones: $\cloneof{O}$ is precisely the set of Boolean functions definable by
$\PL_O$-formulas.
An immediate consequence of Fact~\ref{fact:containment} is that
the set of all propositional languages $\PL_O$ (for $O$ any set of Boolean functions),
ordered by expressive containment,
forms a pre-order, whose induced partial order is isomorphic to the Post lattice. 
Another consequence is:

\begin{fact}
    For every set $O$ of Boolean functions, there is a finite subset $O'\subseteq O$ such that  $\PL_O$ and $\PL_{O'}$ are expressively equivalent.
\end{fact}

It is also well known that $\land$ and $\neg$  form a truth-functionally complete set of operators
for propositional logic. Formally, we
say that $\PL_O$ is \emph{expressively complete}
if every Boolean function $f(x_1, \ldots, x_n)$ (with $n\geq 1$) is 
defined by a formula $\phi(x_1, \ldots, x_n)$ of $\PL_O$ containing no Boolean variables other than $x_1, \ldots, x_n$.

\begin{fact}
    For every set $O$ of Boolean functions,
    $\PL_O$ is expressively complete if and only if $\{\land,\neg\}\preceq O$. Moreover, it is decidable (given a finite set $O$) whether this holds.
\end{fact}

We note that, if expressive completeness were required to hold with respect to Boolean functions $f(x_1, \ldots, x_n)$ with $n\geq 0$ (rather than only for $n>0$),  $\land$ and $\neg$ would not form an expressively complete set of operators, and one would need to further add $\top$ (or $\bot$). 

More generally, it follows from Fact~\ref{fact:downward-closed-decidable} that \emph{every} property of propositional fragments $PL_O$ is decidable, provided that the property is monotone and depends only on $\cloneof{O}$. 

Next, let us consider the question of \emph{succinctness}, which pertains 
to the size of formulas needed to define a given Boolean function. To motivate the question of succinctness, suppose that $O$ and $O'$ are two sets of Boolean functions, such that  $\PL_O$ and $\PL_{O'}$ are expressively equivalent. It does not follow that the two languages are equally succinct. Indeed, if $O'$ contains more operators, potentially it may allow for expressing certain formulas more efficiently. 
If we assume dag-representations,
it is not difficult to see that
there are linear-time translations
between $\PL_O$ and $\PL_{O'}$. 
Indeed, we can simply translate 
a $\PL_O$-formula into a $\PL_{O'}$
formula by replacing each occurrence
of a function in $O$ by its defining 
$\PL_{O'}$-formula. And likewise in the other direction. This yields a
linear blowup in dag-size. The analogous question for tree-size, however is less straightforward. 
Surprisingly, however, Pratt~\cite{Pratt1975} showed that all expressively complete languages have the same succinctness, up to a polynomial:

\begin{theorem}[\cite{Pratt1975}]
\label{thm:pratt}
    Let $O$ and $O'$ be finite sets of Boolean functions such that  $\PL_O$ and $\PL_{O'}$ are expressively complete. Then there are polynomial-time truth-preserving translations between $\PL_O$ and $\PL_{O'}$, even in terms of tree-size.
\end{theorem}

Here, by a \emph{truth-preserving translation} we mean a translation
from $\PL_O$-formulas $\phi$ to
$\PL_{O'}$-formulas $\phi'$, in the same set of variables, such that 
$\sem{\phi}=\sem{\psi}$.
In particular, for instance, Theorem~\ref{thm:pratt} implies that 
$\PL_{\{\land,\neg,\leftrightarrow\}}$-formulas
can be rewritten, in polynomial time (in terms of tree-size), to 
equivalent $\PL_{\{\land,\neg\}}$-formulas. 

Another way to view Theorem~\ref{thm:pratt} is as follows: it tells us that,
whenever $\PL_O$ is expressively complete, then it is maximally succinct, in the following sense: there is no finite set $O'$,
such that  $\PL_{O\cup O'}$ allows for more succinct definitions (up to a polynomial).
 Theorem~\ref{thm:pratt} applies only to expressively complete languages, but some other, similar results have been proved that apply more broadly.
Let us say that a finite set $O$ of Boolean functions is
\emph{robust} if the following holds:
for all finite sets of Boolean functions  $O'$ with $\cloneof{O'}= \cloneof{O}$,  
there is a polynomial-time (in terms of tree-size) truth-preserving translation from $\PL_{O'}$ to $\PL_{O}$. Intuitively,
when $\PL_O$ is robust, it means that
within the class of all languages $\PL_{O'}$
that are expressively equivalent to $\PL_O$, the language $\PL_O$ is maximally succinct (again, up to a polynomial).

\begin{theorem}[\cite{Thomas2012:applicability}]
\label{thm:robust}
Let $O$ be a finite set of Boolean functions. If  $O\preceq \{\land,\top,\bot\}$,
$O\preceq \{\lor,\top,\bot\}$,
$O\preceq \{\oplus,\top,\bot\}$, or
 $\{\land,\lor\}\preceq O$,
then $O$ is robust. Otherwise, $O$ is 
robust as well provided that either
$\land\in O$ or $\lor\in O$.
\end{theorem}

This theorem turns out to be a  helpful tool for proving lower as well as upper bounds on the complexity of various algorithmic problems for certain propositional languages $\PL_O$, as it gives us the freedom to switch to a different base of connectives without impacting the tree-size of formulas up to a polynomial. The first half of Theorem~\ref{thm:robust} tells us that, for certain clones (i.e., $\CloneE$, $\CloneV$, $\CloneL$, or $\CloneM_2$),
it makes no difference which bases are used. For problem classifications with respect to clone restrictions, one often is interested in finding such robust cases, and this is often referred to as \emph{base independence}. 
The notion of robustness is also known under the term of \emph{efficient implementations} \cite[Lem.~4]{DBLP:journals/ijfcs/Schnoor10} or \emph{short representations} \cite[Sect.~3]{DBLP:journals/fuin/CreignouV15}.

\subsection{Complexity of expressibility and minimization}

In this subsection and the following ones, we will consider the computational complexity of different algorithmic tasks pertaining to a language $\PL_O$.
In particular, we will consider
how the complexity depends on the choice of $O$.
Recall that we can represent formulas either as DAGs or as trees. 
Correspondingly, we can study algorithmic problems in two ways. 
The results stated below hold both when complexity is measured with respect to tree-size and when it is measured with respect to DAG-size.

To start, we consider the algorithmic problem of testing expressibility in $\PL_O$.
The \emph{expressibility} problem for $\PL_O$ is as follows:
given a $\PL_{\land,\neg}$-formula $\phi$, is $\phi$ equivalent to a $\PL_O$-formula.

\begin{theorem}[\cite{BohlerSchnoor2007}] Let $O$ be a set of Boolean functions. If $\{\land,\oplus\}\preceq O$ or $\{\lor,\leftrightarrow\}\preceq O$, then the expressibility problem for $\PL_O$ is solvable in polynomial time, otherwise, it is $\class{NP}$-complete.
\end{theorem}

According to the above definition, the input of the expressibility problem is specified as a Boolean formula (using either tree-representation or DAG-representation). Another option would be to assume that it is specified by a truth table. In this case, expressibility is solvable in polynomial time (indeed, in nondeterministic log-space) for \emph{all} languages $\PL_O$ \cite{BergmanSlutzki2000}.

In light of the earlier discussion of succinctness, it is also natural to ask if a given formula can equivalently expressed by
a \emph{short} $\PL_O$-formula.
 The \emph{minimization} problem for $\PL_O$ formalizes this  as follows: given a $\PL_O$-formula $\phi$ and an integer $k$, is $\phi$ equivalent to a $\PL_O$-formula of size at most $k$? 
    The minimization problem for full $\PL$ is in general $\Sigma^p_2$-complete. 
    In \cite{Hemaspaandra2011:minimization}, the authors 
    establish  the following dichotomy (which again applies both with respect to the tree-size of the input formula and with respect to its dag-size).
    \begin{theorem}[\cite{Hemaspaandra2011:minimization}]
    Let $O$ be a set of Boolean functions. If $O\preceq\{\lor,\top,\bot\}$ or $O\preceq\{\land,\top,\bot\}$ or 
    $O\preceq\{\oplus,\top,\bot\}$ then the 
    minimization problem for $\PL_O$ is solvable in polynomial time. Otherwise it is (at least) $\class{coNP}$-hard.
    \end{theorem}
    
\subsection{Complexity of logical reasoning tasks}\label{ssec:complexity-log-reasoning}

We will now consider various logical reasoning tasks. Unless stated otherwise, all results in this section hold both when complexity is measured in terms of tree-size or in terms of DAG-size of the input formula(s).

The \emph{satisfiability problem}
for $\PL_O$ is the problem of deciding, for a 
given $\PL_O$-formula $\phi$, whether there exists a
satisfying truth assignment.

\begin{theorem}[\cite{Lewis_Sat_Problems_for_propositional_calculi}]\label{thm:prop-sat} Let $O$ be a set of Boolean functions. If  $\{\nrightarrow\}\preceq O$, the satisfiability problem for $\PL_O$ is $\class{NP}$-complete. Otherwise, it is solvable in polynomial time.
\end{theorem}

It was later shown~\cite{DBLP:phd/de/Reith2002} that, under tree-reprentation of formulas, 
the above dichotomy can be strengthened to an \class{L}/\class{NP}-complete dichotomy. In other words, the complexity
of the satisfiability problem for a fragment $\PL_O$ is always either in \class{L} or otherwise \class{NP}-complete. Under DAG-representation,
on the other hand, the above dichotomy was refined to a 4-way classification: 

\begin{theorem}[\cite{Reith}]
Let $O$ be a set of Boolean functions. The complexity of the 
satisfiability problem for $\PL{O}$,
under tree-representations of formulas, is
\begin{itemize}
    \item in \class{L} if $O\preceq  \{\lor,\leftrightarrow\}$ or $O\preceq \{\maj, \neg\}$ or $O\preceq \{\neg,\top,\bot\}$
    \item else \class{NL}-complete if 
$O\preceq\{\land,\top,\bot\}$ or $O\preceq \{\lor,\top,\bot\}$ 
    \item else $\oplus$\class{L}-complete if $O\preceq \{\oplus,\top,\bot\}$
    \item else \class{P}-complete if 
     $O\preceq\{\land,\lor,\top,\bot\}$
    \item else \class{NP}-complete (in all remaining cases)
\end{itemize}
\end{theorem}

A refinement of the satisfiability problem that is studied in the fields of parameterized complexity is \emph{weighted satisfiability}. Here, the weight of an assignment is the number of variables that are set to true, and the parameterized problem of interest then is to decide, for a given formula and integer $k$, whether there exists a satisfying assignment of weight exactly $k$. A complete classification in terms of Post's lattice is given in ~\cite{DBLP:journals/fuin/CreignouV15}, distinguishing cases that are $\class{W[SAT]}$-complete (a particular parameterized complexity class~\cite{DBLP:series/txtcs/FlumG06}) or in $\class{P}$. %

Dual to satisfiability is the \emph{tautology problem} for $\PL_O$, which
is the problem of deciding, for a given $\PL_O$-formula $\phi$, whether every truth assignment is
a satisfying truth assignment for $\phi$. By a simple dualization argument, from Theorem~\ref{thm:prop-sat} we can get:

    \begin{theorem}[\cite{Reith,DBLP:phd/de/Reith2002}] \label{thm:prop-taut}
    Let $O$ be a set of Boolean functions. If $\{\to\}\preceq O$
    then the tautology problem for $\PL_O$ is \class{coNP}-complete,
    otherwise it is in polynomial time.     %
    \end{theorem}

The refinements of Theorem~\ref{thm:prop-sat} we described above that hold under tree-representation and under  DAG-representation of formulas also dualize naturally to the tautology problem \cite{Reith,DBLP:phd/de/Reith2002}.

A common generalization of the satisfiability and tautology problems is the problem of \emph{counting the number of satisfying assignments}. Indeed, a formula is satisfiable precisely if this number is positive, and 
a formula is a tautology precisely if this 
number of $2^k$ (where $k$ is the number of propositional variables). By the \emph{counting problem} for $\PL_O$ we will mean the problem
of counting the number of satisfying assignments for a given formula in a given finite set of propositional variables. For full
propositional logic (i.e., $\PL_{\{\land,\neg\}}$) this problem is well known to be complete for the counting complexity class $\sharpp$ (the class of all counting problems that count the number of 
accepting runs of a polynomial-time nondeterministic Turing machine). One can 
also consider the associated decision problems
of \emph{counting modulo 2} (i.e., testing if
the number of satisfying assignments is equal to 1 modulo 2) and the \emph{threshold counting problem} (testing, for a given formula $\phi$
in a given finite set of propositional variables,
and for a given natural number $k$ specified in binary, whether there are at least $k$ satisfying assignments). Dichotomies were obtained in~\cite{Reith} for each of these three problems. We will state here only their 
dichotomy for the pure counting problem.

\begin{theorem}[\cite{Reith}] 
    Given a set of Boolean functions $O$.
    If
    \begin{itemize}
        \item $O\preceq \{\land,\top,\bot\}$, or
        \item $O\preceq\{\lor,\top,\bot\}$, or
        \item $O\preceq \{\oplus,\bot\}$, or
        \item $O\preceq \{\maj, \neg\}$ 
    \end{itemize}
    Then the counting problem for $\PL_O$ is solvable in polynomial time, otherwise
    it is
    $\sharpp$-complete.\label{thm:countingprobPL}
\end{theorem}

The \emph{implication problem} for $\PL_O$
is the problem  of deciding  whether every  truth assignment that satisfies $\phi_1, \ldots, \phi_n$ also satisfies $\psi$, for given $\PL_O$-formulas
$\phi_1, \ldots, \phi_n$ and $\psi$.

    \begin{theorem}[{\cite[Thm.~4.1]{Beyersdorff_2009}}] \label{thm:prop-imp}
    Let $O$ be a set of Boolean functions. The complexity of the implication problem for $\PL_O$ is
    \begin{itemize}
        \item 
        $\class{coNP}$-complete if $\{(x\land y)\lor z\}\preceq O$ or $\{(x\lor y)\land z\}\preceq O$ or $\{\maj \}\preceq O$,
        \item 
        $\class{\oplus L}$-complete, if $\{\oplus\}\preceq O\preceq \{\leftrightarrow, \perp\}$, 
        \item 
        in $\class{AC^0[2]}$ and 
        $\MOD_2$-hard (under suitable reductions), if $\{\neg\}\preceq O\preceq \{\neg, \perp\}$, and
        \item 
        in $\class{AC^0}$ for all other cases. %
    \end{itemize} 
    \end{theorem}

Here, $\MOD_2$ is the language which consists of all words $w\in \{0,1\}^*$ containing an odd number of occurrences of the letter $1$ and $\MOD_2$-hardness means that there is a $\leq^{\class{AC^0}}_m$-reduction from $\MOD_2$ to the problem in question. It is well known result in circuit complexity that this language (which clearly belongs to $\class{AC^0[2]}$) does not belong to $\class{AC^0}$.  In particular, this shows 
that the problem in question does not belong to $\class{AC^0}$.

Also, recall that the implication problem takes as 
    input a set of premises $\phi_1, \ldots, \phi_n$ and a conclusion $\psi$. The more
    restricted 
    ``single-premise'' version of the implication problem was analysed in \cite{Reith}. The dichotomy provided in the former was further refined in~\cite{Beyersdorff_2009}, which we stated above.

Next, the \emph{equivalence problem} for $\PL_O$ is the problem of deciding, for a given pair of $\PL_O$-formulas, whether they have the same set of satisfying truth assignments.

\begin{theorem}[{\cite[Cor.~5.2]{Beyersdorff_2009}}]
    Let $O$ be a set of Boolean functions. The equivalence problem for $\PL_O$ is
        \begin{itemize}
            \item 
            $\class{coNP}$-complete if $\{(x\land y)\lor z\}\preceq O$ or $\{(x\lor y)\land z\}\preceq O$ or $\{\maj \}\preceq O$, 
            \item 
            $\class{AC^0[2]}$-complete (under suitable reductions) if $\{\neg\}\preceq O \preceq \{\neg,\perp\}$, and
            \item 
            in $\class{AC^0}$ in all the other cases.
        \end{itemize}
    \end{theorem}
    The next reasoning problem for $\PL_O$ we consider is the 
    \emph{isomorphism problem}. Given two formulas, $\phi(x_1, \ldots, x_n)$ and $\phi'(y_1, \ldots, y_n)$,
    of $\PL_O$, it asks whether exists a permutation
    $\pi\colon\{x_1,\dots, x_n\}\rightarrow \{y_1,\dots, y_n\}$ such that  $\phi[\pi]$ and $\psi$ have the same set of satisfying truth assignments, where
    $\phi[\pi]$ is obtained from $\phi$ by uniformly replacing each $x_i$ by $\pi(x_i)$.

    \begin{theorem}[\cite{Equiv}]
         Let $O$ be a set of Boolean functions. The isomorphism problem for $\PL_O$ is 
         \begin{itemize}
            \item in $\class{L}$ if 
               $O\preceq \{\lor, \top,\perp\}$ or 
               $O\preceq \{\land, \top,\perp\}$ or 
               $O\preceq \{\oplus, \top \}$,
            \item $\class{coNP}$-hard otherwise.%
             \end{itemize}
    \end{theorem}

To conclude, let us consider what is arguably
the simplest reasoning problem for $\PL_O$, namely the 
 \emph{evaluation problem}. It is the problem of deciding, for a given variable-free $\PL_O$-formula $\phi$, whether $\phi$ evaluates to true.
 When formulas are represented as DAGs, $\PL_{\{\land,\neg\}}$
 is known as the \emph{Boolean circuit value problem (CVP)} and it is well known to be P-complete.
 In the case were formulas are represented as trees, it follows
 from results in \cite{Buss87,BeaudryMcKenzie1995} that the problem is in
$\class{NC_1}$ if $O$ contains only unary and binary Boolean connectives and the formulas are given in infix notation. 
These results were further generalized by Schnoor~\cite{Schnoor05} to the case with arbitrary sets of Boolean functions, of arbitrary arity (using prefix notation).

\begin{theorem}[\cite{DBLP:journals/ijfcs/Schnoor10}]
Let $O$ be a set of Boolean functions containing $\top$ and $\bot$. The evaluation problem for $\PL_O$, when formulas are represented as trees, is 
\begin{itemize}
\item solvable in constant time if $O\equiv \{\top,\bot\}$,
\item $\class{NLOGTIME}$-complete (under suitable reductions) if $O\equiv\{\lor,\top,\bot\}$,  %
\item $\class{coNLOGTIME}$-complete (under suitable reductions) if $O\equiv\{\land,\top,\bot\}$,  %
\item equivalent to $\MOD_2$ (under suitable reductions) if $\{\neg,\top,\bot\}\preceq O\preceq\{\oplus,\top,\bot\}$,  %
\item $\class{NC^1}$-complete (under suitable reductions) otherwise  %
\end{itemize}
\end{theorem}

Here, $\class{NLOGTIME}$ is the class of problems solvable in nondeterministic logarithmic time (cf.~\cite{DBLP:journals/ijfcs/Schnoor10,DBLP:journals/sigact/HemaspaandraV95} for a precise definition). 
A similar analysis was carried out for the DAG case in~\cite{ReithWagner2005}.

\subsection{Teachability and Learnability}\label{section: Teach_and_Learn}

Let $\phi$ be a $\PL_O$-formula. A 
\emph{labeled example} is a
pair $(V,\lab)$, where $V$ is a truth assignment
and $\lab\in\{0,1\}$. We say that $\phi$ 
\emph{fits} this labeled example if $\sem{\phi}_V=\lab$. 
Thus, for instance, let $\phi$ be  the formula $p\land q$, and let
$V_p$ and $V_{pq}$ are the truth assignments that make $p$, respectively $p$ and $q$, true. Then
$(V_p,0)$ and $(V_{pq},1)$ are labeled examples that $\phi$ fits.

Labeled examples can be used in order to illustrate formulas, and, conversely, formulas can be learned from labeled examples. Teachability and learnability refer to these two use cases of labeled examples. We will discuss both.

\subsubsection*{Teachability}

We say that a set of labeled examples
\emph{uniquely characterizes} a $\PL_O$-formula $\phi$ with respect to $\PL_O$ 
if $\phi$ fits the labeled examples and every  $\PL_O$-formula that fits the labeled examples is equivalent to $\phi$. 
As a further refinement, 
we say that a
set of labeled examples
uniquely characterizes  a
$\PL_O[\PROP]$-formula $\phi$ \emph{with respect to $\PL_O[\PROP]$} 
if $\phi$ fits the labeled examples and every $\PL_O[\PROP]$-formula that fits the labeled examples is equivalent to $\phi$. 

Trivially, every formula $\phi\in \PL_O[\PROP]$ is uniquely characterized with respect to 
$\PL_O[\PROP]$ by the set of all
its $2^n$ many labeled examples, where $n=|\PROP|$.
When $O=\{\land,\neg\}$, it is not difficult to see that all $2^n$ labeled examples are needed 
in order to uniquely characterize a formula. 
However, for other sets $O$, less examples may suffice.

\begin{example}
The formula $\phi=p\land q$ is uniquely characterized with respect to $\PL_{\{\land\}}$ by
the set of labeled examples $\{(V_{pq},1)$, $(V_{p},0), (V_{q},0)\}$, where
$V_{pq}$ is the truth assignment that makes
$p$ and $q$ true and all other propositional variables false, and similarly for $V_p$ and $V_q$. The same formula $\phi$ is 
\emph{not} uniquely characterized with respect to 
$\PL_{\{\land,\lor\}}$ by the given set of labeled examples (or \emph{any} set of labeled examples), although it is 
uniquely characterized by these examples
with respect to $\PL_{\{\land,\lor\}}[\{p,q\}]$.
\end{example}

A set of labeled examples that uniquely characterizes a formula $\phi$
is also called a \emph{teaching set}
for $\phi$. Intuitively, if $E$
is a teaching set for $\phi$, then, given $E$ as input, a diligent student (with unlimited computational resources) is able to derive the formula $\phi$ up to  equivalence.

The above example raises the question: how many labeled examples are needed to uniquely characterize a formula with respect to $\PL_O[\PROP]$? 
In particular, when is it the case that polynomially many examples suffice?

\begin{theorem}[from \cite{Dalmau}, cf.~\cite{Arun2024}]
\label{prop:unique1}
Let $O$ be a set of Boolean functions. 
If 
\begin{itemize}
    \item 
    $O\preceq \{\land, \top,\perp\}$ or 
    \item 
    $O\preceq \{\lor, \top,\perp\}$ or 
    \item 
    $O\preceq \{\oplus,\top, \perp \}$
\end{itemize}
then every $\PL_O[\PROP]$-formula (with $\PROP$ a finite set of propositional variables) is uniquely characterized with respect to $\PL_O[\PROP]$ by a 
    set of $|\PROP|+1$ labeled examples.   Otherwise, there
    is no polynomial $p$ such that  every
     $\PL_O[\PROP]$-formula $\phi$ is uniquely characterized with respect to $\PL_O[\PROP]$ by a 
     set of $p(|\PROP|,|\phi|_{\dagsize})$ labeled examples.
     \end{theorem}

The above result has two disadvantages: 
the set of labeled examples it yields
does not necessarily uniquely characterize
$\phi$ with respect to $\PL_O$. It only uniquely
characterizes $\phi$ with respect to $\PL_O[\PROP]$.
In addition, 
the number of examples grows linearly in $|\PROP|$ and hence is not uniformly bounded for a given formula $\phi$.
One may ask when it is possible to uniquely
characterize a formula $\phi$ with respect to $\PL_O$.
The answer to this is given by the following theorems:

\begin{restatable}{theorem}{thmpropunique}
\label{thm:prop-unique}
    Let $O$ be a set of Boolean functions. The following are equivalent:
        \begin{enumerate}
            \item $O\preceq \{\land,\lor, \top,\perp\}$
             or $O\preceq \{\neg, \perp\}$.
             \item Every $\phi\in \PL_O$ is uniquely characterized wrt $\PL_O$ by a finite set of  labeled examples;
             \item There is a function $f$ such that  for all finite sets $\PROP$, every $\phi\in \PL_O[\PROP]$ is uniquely characterized wrt $\PL_O[\PROP]$ by a set of $f(|\phi|_{\dagsize})$ labeled examples.
        \end{enumerate}
\end{restatable}

The proof of Theorem~\ref{thm:prop-unique} is given in Appendix~\ref{AppendixI}.
 Clearly, the  result remains true if $|\phi|_{\dagsize}$ is replaced by $|\phi|_{\tree}$.

\begin{restatable}{theorem}{thmpropuniquepolynomial}
\label{thm:prop-unique-polynomial}
    Let $O$ be a set of Boolean functions. If
    \begin{itemize}
        \item $O\preceq \{\land, \top,\perp\}$ or 
        \item $O\preceq \{\lor, \top,\perp\}$ or 
        \item $O\preceq \{\neg, \perp\}$.
    \end{itemize}
    then every $\PL_O$-formula $\phi$ is uniquely characterized with respect to $\PL_O$
        by a set of at most $|\phi|_{\dagsize}+1$ labeled examples. Otherwise, there is no polynomial $p$ such that  every $\PL_O$-formula $\phi$ is uniquely characterized with respect to $\PL_O$
        by a set of at most $p(|\phi|_{\dagsize})$ many labeled examples. 
\end{restatable}

\begin{proof}
The first part of the theorem
follows directly from results in~\cite{Dalmau}.
For the ``otherwise'' part: 
assume that every $\PL_O$-formula is uniquely characterized with respect to $\PL_O$ by a set of at most 
$p(\dagsize{\phi})$, for some polynomial $p$. It follows
by Theorem~\ref{thm:prop-unique}
that $O\preceq \{\neg ,\perp\}$ or $O\preceq \{\land,\lor,\top ,\perp\}$. Theorem  \ref{prop:unique1} furthermore implies that $O\preceq\{\land,\top,\perp\}$ or $O\preceq\{\lor, \top,\perp\}$ or $O\preceq \{\xor, \top, \perp\}$. A visual inspection of the Post's Lattice reveals that the only $O$ that satisfies these constraints are those where $O\preceq \{\land, \top, \perp\}$, $O\preceq \{\lor, \top, \perp\}$ and $O\preceq \{\neg, \perp\}$. 
\end{proof}

\subsubsection*{Learnability}

Dalmau~\cite{Dalmau} systematically mapped out 
the learnability of $\PL_O$-formulas for different
sets $O$ of Boolean functions, according to 
different notions of learnability.
In order to state the results, we first briefly recall some fundamentals of computational learning theory.

A \emph{concept representation class} is a triple $\mathcal{C}=(C,E,\lambda)$ of a set  $C$ whose elements will be called ``\emph{concept representations}'' (or, \emph{concepts}, for short); a set $E$ whose elements will be called ``\emph{examples}'', and a \emph{valuation function} $\lambda\colon C\rightarrow \mathcal{P}(E)$. 
Two concepts $c, c'$ are said to be equivalent if
$\lambda(c)=\lambda(c')$.
A \emph{graded} concept representation class is a
family of concept representation classes $\mathcal{C}_n$ parameterized by a natural number $n$. 
In computational learning theory, one studies the existence, 
for specific concept classes or graded concept classes, of efficient algorithms for learning concepts from examples. 

We can naturally view $\PL_O$ as a graded concept representation
class, where the parameter $n$ corresponds to the number of propositional variables. More specifically, for $n>0$, we take $\PL_O[\PROP]$
with $\PROP=\{p_1, \ldots, p_n\}$ and 
view it as a concept class
where the concept representation are the
$\PL_O$-formulas over $\PROP$ are the 
examples are all propositional valuations 
for $\PROP$. 
Propositional formulas over
restricted bases such as 
$\{\land\}$ form some of
the most well-studied graded concept representation classes in computational learning theory.

In the \emph{exact learning} framework~\cite{DBLP:journals/ml/Angluin87}, the learning algorithm
has access to an oracle that can answer queries about some unknown target concept, and
the task of the learner is to identify the target concept up to equivalence. One common type of
oracle queries are \emph{membership queries}.
A membership query oracle, $\mathsf{MQ}(c)$ say, receives an input $e\in E$  and
returns ``Yes" if $e\in \lambda(c)$, or ``No" otherwise. A concept class $C$ is \emph{efficiently exactly learnable with membership queries} if there is an algorithm $L$ such that  for every $c\in C$, when $L$ has access to $\mathsf{MQ}(c)$, 
it terminates in time polynomial in the size of $c$, and outputs a concept equivalent to $c$. 
The same definition applies to graded concept classes, where the running time is now required to be bounded by a polynomial in the size of the target concept and $n$. The following results assume DAG-size as the concept size metric.    \footnote{It is stated in~\cite{Dalmau} that these results hold also for tree-size. However, it turns out there is a flaw in the statement and the lower bounds only apply with respect to DAG-size. }

\begin{theorem}[\cite{Dalmau}]
\label{thm:Dalmau}
    Let $O$ be a set of Boolean functions. If
    \begin{itemize}
        \item 
        $O\preceq \{\land, \top, \perp\}$ or
        \item 
        $O\preceq \{\lor, \top, \perp\}$ or \item $O\preceq \{\oplus, \top,\perp\}$
        \end{itemize} 
        then $\PL_O[\PROP]$ is efficiently exactly  learnable with $|\PROP|+1$ membership queries. Otherwise, 
        $\PL_O[\PROP]$ is \emph{not} exactly  learnable with 
        polynomially many membership queries.
\end{theorem}

Note that, while in~\cite{Dalmau}, the lower bound is stated conditionally on  cryptographic assumptions, it holds unconditionally as follows immediately from Theorem~\ref{prop:unique1}, since
efficient exact learnability with membership queries implies the existence of polynomial-sized unique characterizations.
Besides \emph{exact learning with membership queries}, many other notions of learnability have been proposed and studied in the computational learning theory literature. These include
\emph{exact learning with membership and equivalence queries} and the
\emph{Probably-Approximately-Correct (PAC)} learning model. 
We will omit detailed definitions here.
The results in \cite{Dalmau} pertain to learnability in these other learning frameworks as well.
Specifically:

\begin{theorem}[\cite{Dalmau}]
    Let $O$ be a set of Boolean functions. If
    \begin{itemize}
        \item $O\preceq \{\land, \top, \perp\}$ or
        \item $O\preceq \{\lor, \top, \perp\}$ or \item $O\preceq \{\oplus, \top,\perp\}$
    \end{itemize}
    then  $\PL_O[\PROP]$ is polynomial-time exactly learnable with $|\PROP|+1$ equivalence queries, and therefore, also, polynomial-time PAC-learnable. Otherwise, it is not polynomially PAC-predictable with membership queries (and hence also not polynomial-time exactly learnable with membership and equivalence queries) under certain cryptographic assumptions.
\end{theorem}

\section{Fragments of Modal Logics}\label{section:Modal}

In this section, we apply the same general methodology from the previous section to modal logics.

\subsection{Modal Fragments}
\label{sec:modal-fragments}

Fix a countably infinite set of propositional variables ($x_1, x_2, \ldots$).
A \emph{modal formula} is a formula
generated  by the following grammar:
\[ \phi \Coloneqq x \mid \neg\phi\mid\phi\land\psi\mid\phi\lor\psi\mid \Diamond\phi \mid \Box\phi,\]
where $x$ is a propositional variable (from some fixed countably infinite set). We will refer to the set of all modal formulas also as \emph{the full modal language} and we will denote it also by $\ML$.

In order to be able to present some of the results below, we also need to consider \emph{multi-modal formulas}. These are generated by the following grammar:
\[ \phi \Coloneqq x \mid \neg\phi\mid\phi\land\psi\mid\phi\lor\psi\mid \Diamond_i\phi \mid \Box_i\phi,\]
where $x$ is a propositional variable and where $i\in \mathbb{N}$.
 We will refer to the set of all multi-modal formulas also as \emph{the full multi-modal language} and we will denote it also by $\ML^{\omega}$. A
 $k$-modal formula (for  $k\in\mathbb{N}$), is a multi-modal formula containing only modal operators $\Diamond_i$ and/or $\Box_i$ with $1\leq i\leq k$.
 We will denote by $\ML^k$ the 
 set of all $k$-modal formulas.
Note that 
$\ML^k$-formulas are interpreted over Kripke models $M=(W,R_1, \ldots, R_k,V)$ with $k$ accessibility relations and that $\ML^\omega$-formulas are interpreted over Kripke models $M=(W,(R_i)_{i\in\mathbb{N}},V)$ with infinitely many accessibility relations.

We need one more
notation, namely for syntactic \emph{substitution}. Formally, substitution is defined inductively. Let $\sigma={(x_1/\phi_1,\dots x_n/\phi_n)}$ be a partial mapping from propositional variables to  modal formulas. 
Then
\begin{align*}
    x^{\sigma}&\eqdef x \text{ for $x\not \in \{x_1,\dots, x_n\}$.}\\
    x_i^{\sigma}&\eqdef\phi_i \\
    (\varphi \land \psi)^{\sigma} &\eqdef \varphi^{\sigma} \land \psi^{\sigma}\\
    (\varphi \lor \psi)^{\sigma} &\eqdef \varphi^{\sigma} \lor \psi^{\sigma}\\
    (\neg \varphi)^{\sigma}&\eqdef \neg \varphi^{\sigma}\\
    (\Diamond \varphi)^{\sigma}&\eqdef \Diamond \varphi^{\sigma} \\
    (\Box \varphi)^{\sigma}&\eqdef \Box \varphi^{\sigma}
\end{align*}
For an example, $(\Box x\land \Diamond x)^{(x/\Diamond (y\land z))}= \Box \Diamond (y\land z)\land \Diamond \Diamond (y\land z) $.

Now, fix a set of modal formulas $\Phi$.
We denote by $\ML_{\Phi}$ the
smallest set such that 
\begin{enumerate}
    \item Every propositional variable is a formula of $\ML_{\Phi}$, and
    \item If $\phi(x_1, \ldots, x_n)\in\Phi$ and $\psi_1, \ldots, \psi_n\in \ML_\Phi$, then $\oper{\phi}(\psi_1, \ldots, \psi_n)\in\ML_{\Phi}$.
\end{enumerate}

Every $\ML_{\Phi}$-formula $\phi$ can be expanded into an ordinary modal formula $\expand{\phi}\in\ML$ by recursively replacing subformulas of the form
$\oper{\phi}(\psi_1, \ldots, \psi_n)\in\ML_{\Phi}$ by
$\phi^{(x_1/\psi_1, \ldots, x_n/\psi_n)}$. 
Thus, for instance, $\expand{\oper{x_1\land x_2}(p, \oper{\Diamond x}(q))}=p\land\Diamond q$.
\footnote{We may assume a canonical ordering on the set of all propositional variables, so that it is always clear which variable corresponds to the first argument of the operator, etc.}
The expansion operation provides us with a way to assign 
semantics to $\ML_\Phi$-formulas: we let $M,w\models\phi$ if and only if
$M,w\models\expand{\phi}$. 

\begin{example}
\label{ex:modal-fragments}
Let $\Phi = \{\Diamond x, x_1\land x_2\}$. 
Then the expansions of $\ML_\Phi$-formulas
are precisely all modal formulas written
using only $\land$ and $\Diamond$.
In other words, we can think of
$\ML_\Phi$ as defining the fragment of modal logic consisting of formulas using only $\land$ and $\Diamond$. 
Similarly,
\begin{itemize}
    \item For $\Phi=\{\neg x, x\land y, x\lor y, \Diamond x, \Box x\}$, $\ML_\Phi$ defines the full modal language; 
    \item For     $\Phi = \{\top,\bot, x\land y, x\lor y, \Diamond x, \Box x\}$,  $\ML_\Phi$ defines the monotone fragment;
    \item For $\Phi=\{\neg x, x\land y, x\lor y, \CIRCLE x\}$, where 
    $\CIRCLE x$ is short for $\Diamond x\land \Diamond \neg x$,
    $\ML_\Phi$ defines a modal fragment that only allows for expressing the contingency of statements, not their possibility or necessity;
    \item For $\Phi=\{\neg x, x\land y, x\lor y, x\lor \Diamond x, x\land \Box x\}$, $\ML_\Phi$ defines a modal fragment that corresponds to the image of a folklore translation from \logic{S4} to \logic{K4};
    \item For     $\Phi = \{\bot, x\land y, \Box(x\lor y), \Box( x\to y)\}$,  $\ML_\Phi$ defines a modal fragment that corresponds to the image of the G\"odel-McKinsey-Tarski translation from intuitionistic logic to modal logic.
\end{itemize}
\end{example}

To further streamline the notation, we will 
write $\Diamond$ for $\Diamond x$, 
$\Box$ for $\Box x$, $\land$ for
$x\land y$, etc. This allows us, for instance, to
denote the first two fragments in Example~\ref{ex:modal-fragments} by $\ML_{\{\neg,\land,\lor,\Diamond,\Box\}}$
and $\ML_{\{\top,\bot,\land,\lor,\Diamond,\Box\}}$, respectively. 
Note that the choice of propositional variables here is irrelevant, as it does not affect the resulting fragment (which is closed under substitution).

\emph{Multi-modal} fragments 
$\ML^\omega_\Phi$ and $\ML^k_\Phi$ are
defined analogously,
where $\Phi$ can now  be a set of  multi-modal, respectively, $k$-modal formulas.
To keep things simple, we will mostly focus  on the uni-modal case, but
we will also discuss how the definitions and observations  extend to the multi-modal case.

The \emph{size} of a $\ML_\Phi$-formula (or $\ML_\Phi^\omega$ formula or 
$\ML_\Phi^k$ formula)
can again be measured in two ways: using a DAG-style representation or using a tree-style representation. Formally:

\[\begin{array}{lll}
    |x|_{\tree} &=& 1 \\
    |{\oper{\chi}}(\psi_1, \ldots, \psi_n)|_{\tree} &=& 1+\Sigma_i|\psi_i|_{\tree} 
\end{array}\]
whereas $|\phi|_{\dagsize}=|\subf(\phi)|$ where
\[ \subf(\phi) = 
\begin{cases} 
\{x\} & \text{for $\phi$ of the form $x$} \\
\{\phi\}\cup\bigcup_i \subf(\psi_i) & \text{for $\phi$ of the form ${\oper{\chi}}(\psi_1, \ldots, \psi_n)$}
\end{cases}\]

Note that, in both cases, the operators ${\oper{\chi}}$ here are treated as as atomic: they contribute 1 to the size of the formula, regardless of the formula $\chi$ itself. Just as in the propositional case, 
$|\phi|_{\tree}$ can be exponentially larger than 
$|\phi|_{\dagsize}$. This is illustrated by the following example.

\begin{example}
    First, let $\chi(x,y)=(x\land \neg y)\lor (\neg x \land y)$, which defines the exclusive-or function.
Consider the formulas $\phi_n$ inductively defined by
    \[ \phi_1 = {\oper{\chi}}(p,p)\]
    \[ \phi_{n+1} = {\oper{\chi}}(\phi_{n},\phi_n)\]
    It is easy to see that $|\phi_n|_{\dagsize}= n+2$ whereas $|\phi_n|_{\tree} = 2^n-1$.
    
    It is worth pointing out also that
$|\phi|_{\tree}$ can be
exponentially smaller than $|\expand{\phi}|_{\tree}$. Indeed, let $\gamma(x)=\Diamond x\land \Diamond\neg x$, and  
    consider the formulas $\phi_n$ inductively defined by
    \[ \phi_1 = \oper{\gamma}(p)\]
    \[ \phi_{n+1} = \oper{\gamma}(\phi_{n})\]
    It is easy to see that $|\phi_n|_{\dagsize}= |\phi_n|_{\tree} = n+1$. On the other hand,
    $|\expand{\phi_n}|_{\tree} = 5 \cdot 2^{n}-4$. In fact,  every $\ML$-formula equivalent to $\expand{\phi_n}$ is of tree-size at least $2^n$, as was shown in~\cite{Ditmarsch2014:exponential}.
    
\end{example}

We say that $\ML_{\Phi}$ is \emph{expressively
contained} in $\ML_{\Psi}$ with respect to a normal modal logic $\Lambda$ (notation: $\ML_\Phi\preceq^\Lambda\ML_\Psi$) if every formula
$\phi\in\ML_\Phi$ is $\Lambda$-equivalent to a formula
$\psi\in\ML_\Psi$. i.e., $\Lambda\vdash \phi\leftrightarrow\psi$.
Furthermore, we say that $\ML_\Phi$ and $\ML_\Psi$ are
\emph{expressively equivalent} with respect to $\Lambda$ if $\ML_\Phi\preceq^\Lambda\ML_\Psi$ and $\ML_\Psi\preceq^\Lambda \ML_\Phi$, and 
we say that $\ML_\Psi$ is \emph{expressively complete} with respect to $\Lambda$
if $\ML_\Phi\preceq^\Lambda\ML_\Psi$  holds for \emph{every}
set of modal formulas $\Phi$. Clearly, the family of all modal fragments form a pre-ordered set under $\preceq^\Lambda$, and the 
expressively complete fragments are the maximal elements in the pre-order.

\begin{example}
   If we use $\Diamond\Diamond$ as shorthand for the formula $\Diamond\Diamond x$, then
    $\ML_{\{\Diamond\Diamond,\land,\neg\}}$ and $\ML_{\{\Diamond,\land,\neg\}}$ are 
    expressively equivalent with respect to the modal logic $\logic{S4}$, since $\logic{S4}\vdash \Diamond x\leftrightarrow \Diamond\Diamond x$. On the other hand, it is easy to see that $\ML_{\{\Diamond\Diamond,\land,\neg\}} \precneq^{\logic{K}}\ML_{\{\Diamond,\land,\neg\}}$, and hence, in particular, $\ML_{\{\Diamond\Diamond,\land,\neg\}}$ is not expressively complete with respect to $\logic{K}$.
\end{example}

\subsection{Modal Fragments via Modal Clones}
\label{sec:Modal-Fragments-and-Modal-Lattice}

In the case of \emph{propositional}
fragments, we saw that the 
expressive power of a fragment $\PL_O$
is fully determined by the Boolean clone $\cloneof{O}$, and, indeed, $\cloneof{O}$ is precisely the set of all Boolean functions defined by $\PL_O$-formulas. In the case of \emph{modal} fragments, a similar connection to clones exists, but it is less immediate:
propositional semantics identifies each propositional formula with a Boolean function. For modal formulas, on the other hand, 
as observed in early papers by Kuznetsov~\cite{kuznetsov1971functional,kuznetsov1979tools},
there is not one unique such formula-to-function correspondence. It is still possible to  connect modal fragments to clones in a canonical way using the \emph{algebraic semantics} for modal logics, which
we will new briefly recall.
We will focus on uni-modal logics in this 
section, but everything naturally generalizes to the multi-modal case.

We start by recalling that a \emph{normal modal logic} is a set $\Lambda$ of modal formulas that contains all propositional tautologies, is closed under modus ponens and uniform substitution, includes the axiom schema $\Box(\varphi \rightarrow \psi) \rightarrow (\Box \varphi \rightarrow \Box \psi)$ (the K axiom)
and $\Diamond \varphi \leftrightarrow \neg \Box \neg \varphi$ (the duality axiom schema),
and is closed under the necessitation rule, that is, if $\varphi \in \Lambda$ then $\Box \varphi \in \Lambda$.
Examples of normal modal logics include \logic{K} (the minimal normal modal logic, which is complete for the class of all Kripke frames); \logic{K4} (which extends \logic{K} with the axiom $\Diamond \phi\to\Diamond\Diamond \phi$ and is complete for the class of all transitive Kripke frames); \logic{S4} (which further extends \logic{K4} with the axiom $\Diamond \phi\to \phi$ and is complete for the class of all transitive and reflexive Kripke frames);
and \logic{S5} (which further extends \logic{S4} with the axiom
$\phi\to\Box\Diamond \phi$ and is complete for the class of all frames in which the accessibility relation is an equivalence relation).
Another logic that will be featured prominently below is 
\logic{GL}, which extends \logic{K} with the axiom 
$\Box(\Box \phi\to \phi)\to\Box \phi$ and is complete for the class of 
transitive and conversely well-founded frames.  
For completeness results for normal modal logics we refer to 
\cite{Blackburn, CZ97}. 

We will use the notation
$\Lambda\vdash \phi$ as another way of saying that
$\phi\in\Lambda$.

\newcommand{\lland}{\bm{\land}}
\newcommand{\llor}{\bm{\lor}}
\newcommand{\lneg}{\bm{\neg}}
\newcommand{\ltop}{\bm{\top}}
\newcommand{\lbot}{\bm{\bot}}
\newcommand{\lDiamond}{\bm{\Diamond}}

A \emph{modal algebra} is an
algebraic structure $(A,\lland,\llor,\lneg,\ltop,\lbot,\lDiamond)$
consisting of a Boolean algebra
$(A,\lland,\llor,\lneg,\ltop,\lbot)$ extended with an
operator $\lDiamond:A\to A$, satisfying $\lDiamond\lbot=\lbot$ and $\lDiamond (x\llor y)=\lDiamond x\llor \lDiamond y$.
We can associate to every normal modal logic, 
in a canonical way, a modal algebra $\mathcal{A}_\Lambda$, which is known as the \emph{Lindenbaum-Tarski algebra} of $\Lambda$.
Specifically, for each modal formula $\phi\in\ML$,
let $\equivclass{\phi}{\Lambda}$ denote the equivalence class of modal formulas $\{\psi\in\ML \mid \Lambda\vdash \phi\leftrightarrow \psi\}$.
Then 
\[\mathcal{A}_\Lambda=(A_\Lambda,\lland,\llor,\lneg,\ltop, \lDiamond)\] where
$A_\Lambda=\{\equivclass{\phi}{\Lambda}\mid \phi\in\ML\}$ and where
 the operations are defined so that
$\equivclass{\phi}{\Lambda}\lland \equivclass{\psi}{\Lambda}=\equivclass{\phi\land\psi}{\Lambda}$, etc. The reader may verify that (for normal modal logics $\Lambda$) this is indeed a proper definition \cite{Blackburn, CZ97}. 

Under the algebraic semantics, we can
think of each modal formula as defining an operation on $A_\Lambda$: for every modal formula $\phi(x_1, \ldots, x_n)$, we denote by $\sem{\phi}_\Lambda$
the $n$-ary function $f_\phi\colon A_\Lambda^n\to A_\Lambda$ given by
\[f_\phi(\equivclass{\phi_1}{\Lambda}, \ldots, \equivclass{\phi_n}{\Lambda})=\equivclass{\phi^{(x_1/\phi_1, \ldots, x_n/\phi_n)}}{\Lambda}~.\]

We now define a \emph{modal clone}  (with respect to the chosen modal logic $\Lambda$) to be any subclone of $\cloneof{\mathcal{A}_\Lambda}$ (the clone 
generated by the Lindenbaum-Tarski algebra).
It follows from Fact~\ref{fact:clones-lattice} that the set of all modal clones forms a lattice with respect to containment.
We can therefore speak about the \emph{modal clone lattice}.
We  now obtain again a
correspondence between modal fragments and modal clones. For a set $\Phi$ of modal formulas, let $\sem{\Phi}_\Lambda=\{\sem{\phi}_\Lambda\mid\phi\in\Phi\}$.

\begin{fact} Fix a normal modal logic $\Lambda$.
\begin{enumerate}
    \item For all sets $\Phi$ 
of modal formulas,  $\sem{\ML_\Phi}_\Lambda = \cloneof{\sem{\Phi}_\Lambda}$, which is a modal clone.
    \item Every modal clone is equal to $\sem{\ML_{\Phi}}_\Lambda$ for some set of modal formulas $\Phi$.
    \item There is a unique maximal modal clone, namely
    $\sem{\ML_{\{\Diamond,\land,\neg\}}}_\Lambda$.
\end{enumerate}
\end{fact}

In other words, there is a two-way correspondence between modal fragments and modal clones.
Therefore, questions about modal fragments can, in principle, 
be studied through the lens of modal clones. In particular:

\begin{fact} For all normal modal logics $\Lambda$ and sets
$\Phi, \Psi$ of modal formulas, the following are equivalent:
\begin{enumerate}
    \item $\ML_\Phi\preceq^\Lambda\ML_\Psi$
    \item $\cloneof{\sem{\Phi}_\Lambda}\subseteq\cloneof{\sem{\Psi}_\Lambda}$
\end{enumerate}
\end{fact}

Unfortunately, for most normal modal logics $\Lambda$, the modal clone lattice is much less well-structured and much less well-behaved than Post's lattice. Recall from Section~\ref{sec:post} that (i) every Boolean clone is  generated by a finite set of Boolean functions, and hence there are only countably many Boolean clones; (ii) every downward closed set of Boolean clones is the downward closure of a finite set of Boolean clones; and (iii) containment is  decidable for Boolean clones. Each of these properties fails for modal clones:

\begin{fact}
    There are continuum many pairwise-incomparable modal clones with respect to $\logic{K}$. Consequently, not every modal clone with respect to $\logic{K}$ is finitely generated. 
\end{fact}

\begin{proof}
For each $n\in \mathbb{N}$, let $\delta_n$ be shorthand for the modal formula $\Diamond^n\top\land\neg\Diamond^{n+1}\top$, which holds in a world $w$ in a Kripke model precisely if the longest path from $w$ has length $n$.
For each $S\subseteq\mathbb{N}$, let $\Phi_S=\{\delta_{2n}\mid n\in S\}\cup\{\delta_{2n+1}\mid n\not\in S\}$. It is easy to see that $\delta_n$ is not equivalent to any finite Boolean combination of formulas $\delta_m$ with $m\neq n$. It follows that, for $S\neq S'$, it holds that $\cloneof{\sem{\Phi_S}_{\logic{K}}}\not\subseteq \cloneof{\sem{\Phi_{S'}}_{\logic{K}}}$.  
This shows that there are continuum many modal clones and modal clones are, in general, not finite generated. 
We note that the same argument applies to many other normal modal logics, such as $\logic{GL}$.
\end{proof}

\begin{theorem}[\cite{RATSARUSSU+2000+553+570}]
\label{thm:GL-precomplete}
    The (downward closed) set of all non-maximal modal clones for $\logic{GL}$ is not the downward closure of any finite set of modal clones. Indeed, there are infinitely many pre-maximal modal clones with respect to $\logic{GL}$.
\end{theorem}

Here, we call a modal clone \emph{pre-maximal} (with respect to a normal modal logic $\Lambda$) if
it is non-maximal and it is not strictly contained in any other non-maximal modal clone.

By the \emph{containment problem for modal clones with respect to $\Lambda$} we will mean the problem of testing, given finite sets of modal formulas $\Phi$ and $\Psi$, whether $\cloneof{\sem{\Phi}_{\Lambda}}\subseteq \cloneof{\sem{\Psi}_{\Lambda}}$).
It was shown in \cite{Ratsa83,RATSARUSSU+2000+553+570}  that 
the containment problem for modal clones is undecidable for many modal logics, such as \logic{S4} and \logic{GL}, although it is decidable, for instance, for \logic{S5}. In fact, Ra\c{t}\v{a}~\cite{Ratsa83} established a sweeping negative result for extensions of \logic{S4}, which can be strengthened to apply more broadly to extensions of \logic{K4}~\cite{Balder}.

\begin{theorem}[\cite{Ratsa83,Balder}]
\label{thm:expressibility}
Let $\Lambda$ be any normal modal logic extending \logic{K4}. Then 
the containment problem for modal clones with respect to $\Lambda$ is decidable if and only if $\Lambda$ is locally tabular.
\end{theorem}

We recall that a normal modal logic $\Lambda$ is \emph{locally tabular} if, for every finite set of propositional variables, there are only finitely many equivalence classes (w.r.t.~$\Lambda$) of formulas built from these variables. \logic{S5} is an example of a locally tabular logic, whereas most
other common modal logics (such as \logic{K}, \logic{K4},
\logic{S4} and \logic{GL}) are not locally tabular \cite[Section 12.4]{CZ97}.
Theorem~\ref{thm:expressibility} implies that the expressive containment problem for modal fragments is undecidable for all logics $\Lambda$ that are not locally tabular.
Since locally tabular logics are very rare, this can be viewed as a strong negative result. 
For logics that do not extend $\logic{K4}$, little is known about the decidability of the clone containment problem. In particular:
\begin{question}
    Is the containment problem for modal clones with respect to
    $\logic{K}$ decidable?
\end{question}

\subsection{Simple Modal Fragments via Boolean Clones}
\label{sec:simple-modal-fragments}

We have seen that the pre-order of all modal fragments is rather unwieldy for non-locally tabular logics. This motivates further restricting the kinds of modal fragments we consider. We explore this next.
We say that  a set of modal formulas $\Phi$ is \emph{simple} if 
$\Phi$ consists of formulas that may be of the form (i) $\Diamond x$, or (ii) $\Box x$, or (iii) purely propositional. We call $\ML_\Phi$ a \emph{simple modal fragment}
if $\Phi$ is simple.
\begin{example}
   The first two modal fragments in Example~\ref{ex:modal-fragments} are simple, while the other three are not.
\end{example}

We can further distinguish four types of simple sets $\Phi$, according to whether
$\Phi$ contains $\Diamond$ and/or
$\Box$. More precisely, 
for $M\subseteq \{\Diamond,\Box\}$,
we say that a $\Phi$ is 
\emph{$M$-simple} if 
it consists of the modal operators in $M$ plus any number of purely propositional formulas. A modal fragment $\ML_\Phi$ is $M$-simple if
$\Phi$ is $M$-simple. 
In what follows, we will sometimes be
sloppy and consider simple modal fragments of the form $\ML_{M\cup O}$
where $M\subseteq\{\Diamond,\Box\}$ and
$O$ is a set of Boolean functions. This should be understood more precisely as 
$\ML_{M\cup\Phi\}}$ where $\Phi$ contains,
for each Boolean function $f\in O$, an arbitrary-chosen corresponding propositional formula $\phi_f$.

As far as we know, \cite{Gen_Sat_Mod} were the first to systematically study what we call simple modal fragments here.
The results in~\cite{Gen_Sat_Mod}, which pertain to the complexity of the satisfiability problem for such fragments, will be discussed in Section~\ref{sec:simple-complexity}.

For a set $\Phi$ of modal formulas, let $\beta(\Phi)$ denote
the set of
Boolean functions defined by propositional formulas in $\Phi$. The expressive power of a simple modal fragment $\ML_{\Phi}$ is naturally determined by the Boolean clone $\cloneof{\beta(\Phi)}$. However, the relationship is not entirely straightforward: there may be
cases where $\ML_{\Phi}\preceq^\Lambda \ML_{\Psi}$ even though $\cloneof{\beta(\Phi)}$ is not contained in $\cloneof{\beta(\Psi)}$.

\begin{example}
\label{ex:GL}
   Consider the modal logic $\logic{GL}$. In this logic, the following equivalence is valid: 
   \[\logic{GL}\vdash (\Box \phi\lor \Diamond\Box \phi) \leftrightarrow \top\]
   It follows that, in \logic{GL}, $\top$ is definable from 
   $\Diamond,\Box,\lor$. However, 
   note that $\top$ is not definable from $\lor$ alone. Hence, $\ML_{\{\Box,\Diamond,\top\}}\preceq^{\logic{GL}} \ML_{\{\Box,\Diamond,\lor\}}$ even though 
   $\top\not\in\cloneof{\lor}$.
\end{example}

Consider the operation 
$\mathsf{Clos}^\Lambda_M$ on sets of Boolean functions where, for $M\subseteq\{\Diamond,\Box\}$, $\mathsf{Clos}^\Lambda_M(O)=\{f\in \finops{\{0,1\}}\mid \text{$\Lambda\vdash\phi_f\leftrightarrow \phi$ for some $\phi\in\ML_{M\cup O}$}\}$, i.e.,  the set of all Boolean functions 
definable by $\ML_{M\cup O}$-formulas.
Example~\ref{ex:GL} in effect shows that
$\top\in \mathsf{Clos}^{\logic{GL}}_{\{\Diamond,\Box\}}(\{\lor\})$. 

\begin{fact} Fix a consistent normal modal logic $\Lambda$ and a set $M\subseteq\{\Diamond,\Box\}$.~
\begin{enumerate}
    \item For all sets $F$  of Boolean functions,
$\mathsf{Clos}^\Lambda_M(F)$ is a Boolean clone,
\item $\mathsf{Clos}^\Lambda_M$ is 
a closure operation, i.e., 
$F\subseteq \mathsf{Clos}^\Lambda_M(F) = \mathsf{Clos}^\Lambda_M(\mathsf{Clos}^\Lambda_M(F))$.
\item For $M$-simple sets of formulas $\Phi,\Psi$,
we have that
    $\ML_{\Phi}\preceq^\Lambda \ML_{\Psi}$ iff 
    $\mathsf{Clos}^{\Lambda}_{M}(\beta(\Phi))\subseteq \mathsf{Clos}^{\Lambda}_{M}(\beta(\Psi))$.
\end{enumerate}
\end{fact}

This shows that, in order to understand the structure of lattices of simple modal fragments, we can rely on Post's lattice but we need to study the associated closure operations. We will do this next.
To state the next results, we need the following theorem. Because of its connection with Makinson's theorem \cite[Theorem 8.67]{CZ97} stating that every modal logic has as its frame either a frame consisting of a single reflexive point $F_\circ$  or a single irreflexive point $F_\bullet$, we call the next theorem the \emph{Generalized Makinson's theorem}. It gives a more fine-grained classification of logics which have $F_\bullet$ as their frame.  In what follows, $F_\circ$ will denote the Kripke frame that consists of a single reflexive world, and $F_\bullet$ will denote the Kripke frame that consists of a single irreflexive world.
Recall also that for each $n\geq 0$
$$\Diamond^{\leq n}\phi = \bigvee_{k\leq n} \Diamond^k \phi,$$
where $\Diamond^0 \phi = \phi$ and $\Diamond^{k+1}\phi = \Diamond (\Diamond^k \phi)$. 
Also $\Box^0 \phi = \phi$ and $\Box^{k+1}\phi = \Box (\Box^k \phi)$.

\begin{theorem}[Generalized Makinson's theorem]\label{GMT}
    Every normal modal logic $\Lambda$, 
    is of exactly one of the following types:
    \begin{description}

\item[Type A:] $F_\circ\models\Lambda$,  %

\item[Type B:] $F_\bullet \models \Lambda$ and $\Lambda\vdash \Box^n \bot$ for some $n\geq 1$, 

\item[Type C:]  $F_\bullet \models \Lambda$, $\Lambda\not\vdash \Box^n \bot$ for all $n\geq 1$, and
$\Lambda\vdash\Diamond^{\leq n}\Box\bot$ for some $n\geq 1$.
    \end{description}
\end{theorem}

\begin{proof}
It is easy to see that the three types are mutually exclusive. Therefore, it suffices to show every normal logic logics is of one of the three types. The proof will make use of some basic notions and constructions from the topo-algebraic semantics of modal logics. We refer to \cite{Blackburn, CZ97, Esakia19, Venema04} for details on these notions.
Suppose $\Lambda$ is not of Type B and not of Type C. We show that it is of Type A. Since $\Lambda$ is not of Type A or B, we have that $\Lambda\not\vdash \Box^n \bot$ for all $n\geq 1$ and $\Lambda\not\vdash \Diamond^{\leq n}\Box \bot$ for all $n\geq 1$. We want to show that $F_\circ$ is a $\Lambda$-frame. Then the set of formulas $\{\neg \Diamond^{\leq n}\Box \bot: n\geq 1\}$ is $\Lambda$-consistent. This implies (by Lindenbaum's Lemma) that there is a $\Lambda$-maximal consistent set $\Gamma$ such that  $\Diamond^{\leq n}\Box \bot\notin \Gamma$ for each $n\geq 1$ (see e.g., \cite[Chapter 4]{Blackburn}). By the Truth Lemma of the canonical models, this implies that $\mathcal{M}_\Lambda, \Gamma\not\models \Diamond^{\leq n}\Box \bot$ for each $n\geq 1$. But this means that for each $n\geq 1$ and each sequence $\Gamma_0, \dots, \Gamma_n$ such that  $\Gamma_0 = \Gamma$ and $\Gamma_i R \Gamma_{i+1}$, there is a point $\Delta$ such that  
$\Gamma_n R \Delta$. So each point in $R^\omega [\Gamma] = \{\Delta: \exists \Gamma_0, \dots \Gamma_n, \Gamma_0 = \Gamma, \Gamma_n = \Delta\  \& \ \Gamma_i R \Gamma_{i+1}$ for each $i\leq n-1\}$ has a successor. So it is easy to see that $F_\circ$ is a p-morphic image of 
$R^\omega [\Gamma]$.
    
It is well known that the canonical model can be equipped by a topology generated by the sets of the form $\llbracket \phi \rrbracket = \{ \Gamma : \Gamma \models \varphi\}$ \cite{Blackburn, Esakia19, Venema07}. Note that as the complement of the set $\llbracket \phi \rrbracket$ equals $\llbracket \neg \phi \rrbracket$, each such set is clopen, i.e., both closed and open. Moreover, it is also well known that generated subframes (that is, subsets $U$ such that  $x\in U$ and $xR y$ imply $y\in U$), which are closed in this topology, preserve the validity of formulas. That is, every closed generated subframe is a (descriptive general) frame of $\Lambda$ \cite{Esakia19, Venema04}. Now we consider the topological closure $\overline{R^\omega [\Gamma]}$ of the set $R^\omega [\Gamma]$. Clearly, by definition, it is topologically closed.  By \cite{Venema04} this set is an $R$-generated subframe.  Hence $\overline{R^\omega [\Gamma]}$ is a $\Lambda$-frame. 

Now suppose $x\in \overline{R^\omega [\Gamma]}$. If $x\in R^\omega [\Gamma]$, then as we observed, $x$ has a successor. Now suppose $x\in \overline{R^\omega [\Gamma]} - R^\omega [\Gamma]$. Suppose $x$ has no successors. Then $x\models \Box\bot$. But then $\llbracket \Box\bot \rrbracket $, the clopen set defined by the formula $\Box \bot$, is a clopen set containing $x$ having empty intersection with $R^\omega [\Gamma]$. This is a contradiction since 
$x$ belongs to the topological closure of $R^\omega [\Gamma]$. Thus, $x$ has a successor and every point in $\overline{R^\omega [\Gamma]}$ has a successor. Now it is easy to see that by mapping every point of $\overline{R^\omega [\Gamma]}$ onto $F_\circ$ we obtain that $F_\circ$ is a p-morphic image of $ \overline{R^\omega [\Gamma]}$, which implies that $F_\circ$ is a p-morphic image of a $\Lambda$-descriptive frame 
$\overline{R^\omega [\Gamma]}$. Therefore, $F_\circ$ is a $\Lambda$-frame. So $\Lambda$ is of Type A.
\end{proof}

Almost all common normal modal logics, such as \logic{K}, are of Type A. A notable exception is
\logic{GL}, which is of Type C. Indeed,
$\logic{GL}\vdash \Box\bot\lor\Diamond\Box\bot$. Logics of Type B tend to not be very interesting. In particular, every logic of Type B is locally tabular.

\begin{remark}
{\em We observe that above ${\bf K4}$ there exists the least logic, which is not of Type A. This will be the logic $L_\circ$ axiomatized above ${\bf K4}$ by the so-called Jankov-Fine formula of $F_\circ$ \cite[Section 3.4]{Blackburn} and \cite[Chapter 9]{CZ97} (note that such formulas can be defined above ${\bf K4}$, but not above ${\bf K}$). The Jankov-Fine formula of a rooted frame $F$ exactly axiomatizes the least logic which does not have $F$ as its frame. 

We also note that this logic $L_\circ$ is not ${\bf GL}$. To see this, consider the frame $G$ shown in Figure~\ref{fig:G}, consisting of a descending transitive chain of irreflexive points, e.g., the natural numbers with $>$, together with an extra reflexive point $\omega$, which ``sees'' all the other points. Then clearly $G\not\models {\bf GL}$ as it contains a reflexive point. However, it is easy to see that finite rooted bounded morphic images of generated subframes of  $G$ are finite chains of irreflexive points. This, means that the Jankov-Fine formula of $F_\circ$ is valid on $G$ as $F_\circ$ is not a bounded morphic image of a generated subframe of $G$; we refer the reader to  \cite[Chapter 9]{CZ97} for the theory of frame-based formulas in modal logic. So $G$ is a frame of this logic and not a ${\bf GL}$-frame, which means that $L_\circ$ and ${\bf GL}$ are different logics. 
It is also easy to see that the frame $G$ refutes $\Box^n \bot$ for any $n\in \omega$. Therefore,  $L_\circ$ is  not of Type B and hence, by Theorem~\ref{GMT}, it is of Type C. }
\end{remark}

\begin{figure}[t]
\[
\begin{tikzcd}[column sep=1.8em]
\omega \arrow[loop above] \arrow[r, dashed] & \cdots \arrow[r] &
3 \arrow[r] \arrow[rr, bend left=25] \arrow[rrr, bend left=35] &
2 \arrow[r] \arrow[rr, bend left=25] &
1 \arrow[r] &
0
\end{tikzcd}
\]

\caption{The frame $G=(\mathbb{N}\cup\{\omega\},{>}\cup\{(\omega,\omega)\})$}
\label{fig:G}
\end{figure}
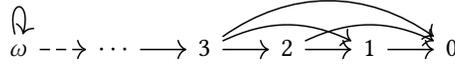

\begin{theorem}\label{thm:simple-to-propositional}
    Let $\Lambda$ be any consistent normal modal logic, let 
    $M\subseteq\{\Diamond,\Box\}$, and let $C$ be any Boolean clone.
    \begin{enumerate}
        \item If $\Lambda$ is of Type A, then $\mathsf{Clos}^{\Lambda}_{M}(C)=C$, i.e.,
         $\mathsf{Clos}^{\Lambda}_{M}$ is the identity function on Boolean clones.

        \item If $\Lambda$ is of Type B, then 
         $\mathsf{Clos}^{\Lambda}_{M}(C)=
           \begin{cases}
           C & \text{if $M=\emptyset$} \\
           C\sqcup\{\bot\} & \text{if $M=\{\Diamond\}$} \\
           C\sqcup\{\top\} & \text{if $M=\{\Box\}$} \\
           C\sqcup\{\top,\bot\} & \text{if $M=\{\Diamond,\Box\}$} \\
           \end{cases}$
           \item If $\Lambda$ is of Type C, then
           $C\subseteq \mathsf{Clos}^{\Lambda}_{M}(C) \subseteq \begin{cases}
           C & \text{if $M=\emptyset$} \\
           C\sqcup\{\bot\} & \text{if $M=\{\Diamond\}$} \\
           C\sqcup\{\top\} & \text{if $M=\{\Box\}$} \\
           C\sqcup\{\top,\bot\} & \text{if $M=\{\Diamond,\Box\}$} \end{cases}$
        \end{enumerate}
    \end{theorem}
\begin{proof}
Type A: 
Let $C$ be any Boolean clone and let $f\in \mathsf{Clos}^{\Lambda}_{M}(C)$, where $M\subseteq\{\Diamond,\Box\}$.
By definition, 
there is a $\chi\in\ML_{M\cup \Phi}$ such that 
$\Lambda\vdash\phi_f\leftrightarrow \chi$, where $\Phi=\{\phi_f\mid f\in C\}$. It follows that
$F_\circ\models\phi_f\leftrightarrow\chi$. 
Let $\chi'$ be obtained from $\chi$
by removing all modal operators (i.e.,
recursively replacing subformulas of
the form $\Diamond\theta$ or $\Box\theta$
by $\theta$). Then we have that
$F_\circ\models \chi\leftrightarrow \chi'$.
Therefore, we have that 
$F_\circ \models\phi_f\leftrightarrow \chi'$. Observe that both $\phi_f$ and $\chi'$ are propositional formulas.
It follows that $\chi'$ and $\phi$ are
equivalent as propositional formulas.
Since $\chi'$ is composed of operations corresponding to formulas in $\Phi$, 
it follows that $f$ belongs to the Boolean clone $C$.

Type B:  The right-to-left inclusions follow directly from the fact that
$\Lambda\vdash \top\leftrightarrow\Box^n x$ and $\Lambda\vdash\bot\leftrightarrow \Diamond^n x$. For the left-to-right inclusions, we reason as follows. Since $\Lambda$ is of Type B, we have that $F_\bullet\models\Lambda$.
Now, let $C$ be any set of Boolean functions and let $f\in \mathsf{Clos}^{\Lambda}_{M}(C)$, where $M\subseteq\{\Diamond,\Box\}$.
By definition, 
there is a $\chi\in\ML_{M\cup \Phi}$ such that 
$\Lambda\vdash\phi_f\leftrightarrow \chi$, where $\Phi=\{\phi_f\mid f\in C\}$. It follows that
$F_\bullet\models\phi_f\leftrightarrow\chi$. 
Let $\chi'$ be obtained from $\chi$ by replacing all subformulas of the form $\Diamond\theta$ by $\bot$ and all subformulas of the form $\Box\theta$ by $\top$. Then we have that
$F_\bullet\models\chi\leftrightarrow\chi'$.
Therefore, we have that $F_\bullet\models\chi'\leftrightarrow \psi$.
Observe that both $\phi_f$ and $\chi'$ are propositional formulas.
It follows that $\chi'$ and $\phi$ are
equivalent as propositional formulas.
Since $\chi'$ is composed of operations corresponding to formulas in $\Phi$, 
as well as $\bot$ if $\Diamond\in M$ 
and $\top$ if $\Box\in M$,
$f$ belongs to the respective Boolean clone as stated in item 2 of the theorem.

 Type C: The same argument applies as for the left-to-right inclusions in the case of logics of Type B. (Note, however, that the proof of the right-to-left inclusions does not applies for Type C logics.)
\end{proof}

The above results imply that, for all normal modal logics $\Lambda$, the simple modal fragments form a pre-ordered set with a clear structure, which (just as in the case of propositional fragments) is derivable from the structure of Post's lattice. In particular:

\begin{corollary}
    Fix a normal modal logic $\Lambda$. 
\begin{enumerate}
    \item 
    There are only countably many simple modal fragments $\ML_\Phi$, up to expressive equivalence. 
    \item 
    There are no infinite antichains of simple modal fragments, 
    i.e, families of simple modal fragments that are pairwise incomparably with respect to $\preceq^\Lambda$.
    \item Every simple modal fragment $\ML_\Phi$  is expressively equivalent with respect to $\Lambda$ to $\ML_\Psi$ for some finite
    subset $\Psi\subseteq\Phi$. 
    \item 
    It is decidable, given a simple set $\Phi$, whether
    $\ML_\Phi$ is expressively complete with respect to $\Lambda$. 
\end{enumerate}
\end{corollary}

\begin{proof}
    The first three items follow directly from the above results together with known properties of Post's lattice discussed in Section~\ref{sec:post}. For the fourth item, we can argue as follows:
    fix $M\subseteq\{\Diamond,\Box\}$, and let $\mathcal{F}$ be the set of all Boolean clones $C$ for which it
    holds that $\ML_{M\cup C}$ is \emph{not} expressively complete with respect to $\Lambda$. Then $\mathcal{F}$ is a downward closed set and hence, by Fact~\ref{fact:downward-closed-decidable} is decidable.
\end{proof}

For normal modal logics of Type A or B, we furthermore get:

\begin{corollary}\label{cor:containment-simple-decidable}
    Fix a normal modal logic $\Lambda$ of Type A or B, and $M\subseteq\{\Diamond,\Box\}$. Given two finite $M$-simple sets $\Phi, \Psi$, it is decidable whether $\ML_{\Phi}\preceq^\Lambda \ML_{\Psi}$.
\end{corollary}

We leave it as an open problem whether a similar result holds for Type C logics.

\paragraph{Simple multi-modal fragments}
The above definitions and observations all 
generalize in a natural way to multi-modal fragments, as long as care is taken to define \emph{simple sets} in the appropriate way. We consider a set of $\ML^k$-modal formulas to be \emph{simple} if the following three conditions hold:
\begin{itemize}
    \item $\Phi$ consists of purely propositional formulas and/or formulas of the form $\Diamond_i x$ or $\Box_i x$,
    \item For all $i,j\leq k$, $\Diamond_i x\in \Phi$ iff $\Diamond_j x \in\Phi$, and
    \item For all $i,j\leq k$, $\Box_i x\in \Phi$ iff $\Box_j x \in\Phi$.
\end{itemize}
and similarly for $\ML^\omega$.
With a slight abuse of notation, we 
will say that a set of $\ML^k$-formulas, respectively, $\ML^\omega$-formulas $\Phi$ is
$M$-simple, for $M\subseteq\{\Diamond,\Box\}$, if $\Phi$ is simple and
$\Diamond_i x\in\Phi$ precisely if $\Diamond\in M$, and $\Box_i x\in\Phi$ precisely if $\Box\in M$. With this definition of simplicity, all the above facts remain true for multi-modal fragments of the form $\ML^\omega_\Phi$ or $\ML^k_\Phi$ (provided that $k\geq 1$).

\subsection{Expressive power and succinctness}
\label{sec:modal-expressive-power}

We will now briefly recap the main facts regarding expressive power
of modal fragments that follow from the connections to modal clones discussed in Section~\ref{sec:Modal-Fragments-and-Modal-Lattice}. Before we so so, we introduce some relevant 
notions and algorithmic problems.
We call a modal fragment $\ML_\Phi$ \textit{pre-complete} in $\Lambda$ if 
$\ML_\Phi$ is not expressive complete, and
for every modal fragment $\ML_\Psi$ such that  $\ML_\Phi$ is strictly expressively contained in $\ML_\Psi$, it holds that $\ML_\Psi$
is expressively complete.

\begin{theorem}
\label{fact:modal-expressive-completeness}
    Let $\Lambda$ be any normal modal logic. The following hold with respect to $\Lambda$:
    \begin{enumerate}
        \item For all sets of modal formulas $\Phi$,  $\ML_\Phi$ is expressively complete iff $\{\Diamond,\land,\neg\}\preceq^\Lambda\Phi$.
        \item  There are continuum many pairwise expressively incomparable modal fragments $\ML_\Phi$. 
        \item There are infinitely many expressively pre-complete modal fragments with respect to $\logic{GL}$ \cite{RATSARUSSU+2000+553+570}.
    \end{enumerate}
\end{theorem}
Further, consider the following two algorithmic problems: 
 \begin{align*}
    &\underline{\textbf{Expressibility Problem (for $\Lambda$)}}:\\
    &\textit{Input:} \text{ A modal formula $\varphi$ and finite set of modal formulas $\Phi$}.\\
    &\textit{Output:} \text{ Yes if $\varphi$ is $\Lambda$-equivalent to a $\ML_\Phi$-formula; no otherwise}.\\
    & \text{  } \\
    &\underline{\textbf{Expressive Completeness Problem (for $\Lambda$)}}:\\
    &\textit{Input:} \text{ A finite set of modal formulas $\Phi$}.  \\
    &\textit{Output:} \text{ Yes if $\ML_{\Phi}$ is expressively complete; no otherwise}.
\end{align*}

Note that the expressibility problem is equivalent to the problem of
testing, for two given sets of modal formulas $\Phi, \Psi$, whether
$\ML_\Phi\preceq^\Lambda \ML_\Psi$.
Furthermore, note that the expressive completeness 
reduces to the expressibility problem, due to Fact~\ref{fact:modal-expressive-completeness}.
The results in Section~\ref{sec:Modal-Fragments-and-Modal-Lattice} imply:

\begin{theorem}
\label{thm:expressibility-again}
Let $\Lambda$ be any modal logic 
such that  $\logic{K4}\subseteq \Lambda$. 
If  
$\Lambda$ is locally tabular,
the expressibility problem and the expressive completeness problem are decidable for $\Lambda$. 
Otherwise, the expressibility problem is undecidable for $\Lambda$.
\end{theorem}

For non-transitive modal logics (i.e., logics that are not above \logic{K4}), little is known with regards to the above algorithmic problems.

For \emph{simple} modal fragment the situation is quite different. Indeed, it follows from the results in Section~\ref{sec:simple-modal-fragments} that the expressive completeness problems is decidable for \emph{every} normal modal logic, when restricted to simple modal fragments; and the expressibility problem
is decidable for every normal modal logic of Type A or B when restricted to $M$-simple modal fragments, for $M\subseteq\{\Diamond,\Box\}$.

Next, let us consider succinctness.
Recall that all expressively complete propositional fragments have the same
succinctness up to a polynomial 
(i.e., there are polynomial 
truth-preserving translations between
each pair of such fragments, with respect to tree-size).
A remarkable recent result in \cite{Berkholz2024s} states that, 
in the case of the modal logic $\logic{K}$, there are precisely
two classes of expressively complete simple modal fragments, when it 
comes to succinctness.

\begin{theorem}[\cite{Berkholz2024s}]
\label{thm:Berkholz}
Let $\Phi$ be any simple set of modal formulas.
If $\ML_{\Phi}$ is expressively complete with respect to \logic{K}, then
exactly one of the following holds:
\begin{itemize}
    \item There are polynomial-time truth-preserving translations between $\ML_{\Phi}$
    and $\ML_{\{\Diamond,\land,\neg,\top\}}$ with respect to \logic{K}.
    \item There are polynomial-time truth-preserving translations between $\ML_{\Phi}$
    and $\ML_{\{\Diamond,\land,\neg,\top,\leftrightarrow\}}$ with respect to \logic{K}, in terms of tree-size.
\end{itemize}
Moreover, $\ML_{\{\Diamond,\land,\neg,\top,\leftrightarrow\}}$ is exponentially more succinct than $\ML_{\{\Diamond,\land,\neg,\top\}}$ with respect to \logic{K}, in terms of tree-size.
\end{theorem}

The above theorem pertains to expressively complete \emph{simple} modal fragments, and it implies that, among these, there are ones that are ``maximally succinct''. It is an interesting question whether the same
holds for not-necessarily-simple modal fragments:

\begin{question}
    Is there a finite set of modal formulas
    $\Phi$ such that  for each finite set of modal formulas $\Psi$,
    there is a polynomial-time truth-preserving translation
    from $\ML_\Psi$ to $\ML_\Phi$?
\end{question}

For example, if the basic modal language is extended with a ``contingency'' operator
($\Diamond x \land \Diamond \neg x$) it is known to become exponentially more succinct~\cite{Ditmarsch2014:exponential}. This falls outside of the regime of Theorem~\ref{thm:Berkholz} since $\Diamond x\land\Diamond\neg x$ is not simple. 

\subsection{Complexity of logical reasoning tasks for simple modal fragments}
\label{sec:simple-complexity}

We will now review several results pertaining the algorithmic reasoning tasks. These results pertain only to simple modal fragments, as this is the
only setting in which useful techniques are known for proving such results. Some of the results
pertain to multi-modal logics.
For a normal modal logic $\Lambda$, we denote by $\Lambda_k$ the normal $k$-modal logic containing a copy of all axioms of $\Lambda$ for each of the $k$ modal operators. Thus, for instance, 
$\logic{K}_2$ is the bi-modal logic of all Kripke models, and $\logic{S5}_2$ is the bi-modal logic of Kripke models in which both accessibility relations are equivalence relations.

A general exposition on the complexity of modal logic variants and their fragments can be found in \cite{DBLP:books/daglib/0030205}.

\paragraph{Consistency and Provability}
By the \emph{$\Lambda$-consistency problem} for 
$\ML_{\Phi}$ with respect to a normal modal logic $\Lambda$, we mean the problem of deciding
whether a given $\ML_{\Phi}$-formula $\phi$ is $\Lambda$-consistent. This the modal equivalent of the propositional satisfiability problem and (for modal logics $\Lambda$ that are Kripke-complete) it coincides with the modal satisfiability problem.

\begin{theorem}[\cite{Gen_Sat_Mod}]
\label{thm:gensatmod}
    Let $\ML_\Phi$ be a simple modal fragment containing at least one of $\Diamond,\Box$. 
    \begin{enumerate}
        \item[1.] If $\{x\land\neg y \}\preceq \Phi$ or $\{x\land (y\lor z), \bot,\Diamond,\Box\}\preceq \Phi$,  \logic{K}-consistency   for $\ML_{\Phi}$ is $\class{PSPACE}$-complete.        \item[2.] If $\{\land, \Diamond,\Box\}\preceq \Phi\preceq \{\land, \top,\perp, \Diamond,\Box\}$, \logic{K}-consistency for $\ML_{\Phi}$ is 
         $\class{coNP}$-complete. 
    \item[3.] Otherwise, %
 \logic{K}-consistency for $\ML_{\Phi}$ is solvable in polynomial time, 
    \end{enumerate}
    The same holds for  simple multi-modal fragments and for simple $k$-modal fragments $(k\geq 1)$. 

\end{theorem}

A similar classification result is provided in \cite{Gen_Sat_Mod} for
(multi-modal versions of) \logic{KD}. Furthermore, nearly complete classifications are provided for \logic{T}, \logic{K4}, \logic{S4}, and \logic{S5}, leaving open only the cases where $\{\oplus\}\preceq \Phi \preceq\{\oplus,\bot,\Diamond,\Box\}$. In the  
$\logic{S5}_k$ case, unlike in the general case, the complexity of the consistency problem sometimes changes depending on the value of $k$. 
It is also explained in \cite{Gen_Sat_Mod} how such classifications for consistency problems give rise (by dualizing them) to classifications for the analogous \emph{provability} problems (or, equivalently, for Kripke-complete modal logics, \emph{validity} problems).

\paragraph{Reasoning with TBoxes}
Next, we describe some knowledge-representation related 
reasoning tasks studied in~\cite{Gen_Sat_ALC}. These were orginally
defined in terms of description logics, but we will recast them here  equivalently in terms of (multi-)modal logic.
An $\ML_{\Phi}$-\emph{implication} is an implication $\phi\to\psi$ where 
$\phi,\psi$ are $\ML_{\Phi}$-formulas. 
An \emph{$\ML_{\Phi}$-TBox} is a finite set of $\ML_{\Phi}$-implications.

\begin{example}
    An example of a $\ML_{\{\Diamond,\lor\}}$-TBox is the set $\{\Diamond (p_1\lor p_2)\to p_3, ~~ (p_2\lor p_3)\to (p_4\lor \Diamond p_5)\}$. Notice that $\Diamond$ binds stronger than $\to$, yielding $(\Diamond(p_1\lor p_2))\to p_3$ for the first implication.
\end{example}

 A Kripke model $M$ \emph{globally satisfies} (or, simply, \emph{satisfies})
a TBox $T$, written $M\models T$, if for all worlds $w$, 
and for all  $(\phi\to\psi)\in T$, it holds that
$M,w\models\phi\to\psi$.

A number of knowledge-representation related reasoning tasks are studied in~\cite{Gen_Sat_ALC}. We highlight
here only three of these.
By the \emph{TBox-satisfiability problem} for $\ML_{\Phi}$ we mean the following problem:
given an $\ML_{\Phi}$-TBox, does there exist a Kripke model 
$M$ such that  $M\models T$.
By the \emph{TBox-subsumption problem} (also see \cite{DBLP:journals/corr/abs-1205-0722,meier2014alcsubsumption})
for $\ML_{\Phi}$ we mean the following problem:
given an $\ML_{\Phi}$-TBox $T$ and an $\ML_{\Phi}$-implication $\phi\to\psi$, is it the case that, 
for every model $M$, we have that $M\models T$ implies $M\models\phi\to\psi$? 
Finally, the \emph{satisfiability problem under TBoxes} for $\ML_{\Phi}$ is the problem, 
given an $\ML_{\Phi}$-TBox $T$ and a $\ML_{\Phi}$-formula $\phi$, whether there is a model $M\models T$ and a world $w$ such that  $M,w\models\phi$.

These capture fundamental reasoning tasks in knowledge representation. 
The same definitions apply also to multi-modal fragments $\ML^\omega_{\Phi}$. Indeed, this is the more relevant case in knowledge representation, since the description logic $\mathcal{ALC}$ corresponds to the full multi-modal language
$\ML^\omega$
while, for example, the description logic $\mathcal{EL}$ corresponds to the multi-modal fragment
$\ML^\omega_{\{\land,\top,\Diamond\}}$.

For the TBox-satisfiability problem with respect to the multi-modal logic $\logic{K}_\omega$, the classification 
obtained in~\cite{Gen_Sat_ALC} is summarized in Figure~\ref{fig:TSAT}. For
\emph{TBox subsumption} and \emph{satisfiability under TBoxes}, similar clasifications are obtained in~\cite{Gen_Sat_ALC,DBLP:journals/corr/abs-1205-0722,meier2014alcsubsumption}, 
which we omit here.
\footnote{Note that, regardless of the choice of $\Phi$, TBox satisfiability reduces to TBox subsumption (since $T$ is unsatisfiable iff
$T$ entails $p\to q$ for fresh propositional variables $p,q$) and to 
satisfiability under a TBox (since $T$ is satisfiable iff $p$ is satisfiable under $T$, for a fresh propositional variable $p$). Also see \cite[p.~57]{Gen_Sat_ALC} for further problems and interreducibilities among them.}

\begin{figure}
\resizebox{\linewidth}{!}{\begin{tabular}{@{}rcccc@{}}\toprule
& $M=\emptyset$ & $M=\{\Diamond_i\mid i\in\mathbb{N}\}$ & $M=\{\Box_i\mid i\in\mathbb{N}\}$ & $M=\{\Diamond_i,\Box_i\mid i\in\mathbb{N}\}$ \\
\midrule
$\class{EXPTIME}$-complete if \ldots 
& -- 
& $\{\lor,\top,\bot\}\preceq O$ or 
& $\{\land,\top,\bot\}\preceq O$ or & $\{\top,\bot\}\preceq O$ or \\
&   
& $\{\neg\}\preceq O$ 
& $\{\neg\}\preceq O$ 
& $\{\neg\}\preceq O$\\\\
$\class{NP}$-complete if \ldots
& $\{\land,\lor,\top,\bot\}\preceq O$ or
& -- & -- & --\\
& $\{x\oplus y\oplus z\oplus \top\}\preceq O$ \\\\
$\class{P}$-complete if \ldots       
& $O\equiv\{\land, \top,\perp\}$ or 
& $\{\top,\bot\} \preceq O\preceq$ 
& $\{\top,\bot\}\preceq O\preceq$  
& --\\
& $O\equiv\{\lor, \top,\perp\}$ 
& $\{\land,\top,\bot\}$ 
& $\{\lor,\top,\bot\}$ \\\\
$\class{NL}$-complete if \ldots 
& $\{\neg\}\preceq O\preceq\{\neg,\top,\bot\}$ or
& -- & -- & --\\
& $O\equiv \{\top,\bot\}$\\ \\
Trivial \ldots
& otherwise 
& otherwise 
& otherwise 
& otherwise\\\bottomrule
\end{tabular}}
\caption{Complexity of TBox satisfiability with respect to the logic $\logic{K}_\omega$ for simple multi-modal fragments $\ML^\omega_{M\cup O}$, where $M\subseteq\{\Diamond_i,\Box_i\mid i\in\mathbb{N}\}$ and $O$ consists of purely propositional formulas. Also see \cite[Tab.~3]{Gen_Sat_ALC}.}
\label{fig:TSAT}
\end{figure}

\subsection{Teachability and learnability for simple modal fragments}

Recall that, in the setting of (classical) propositional logic, a labeled example was a pair $(V,\lab)$ where $V$ is a truth assignment and $\lab\in\{1,-\}$. In the 
setting of a modal logic $\Lambda$, it is
natural to consider a labeled example to be of the form $(M,w,\lab)$ where $M$ is a finite Kripke model based on a $\Lambda$-frame, $w$ is a world of $M$, and
$\lab\in\{1,0\}$. A modal formula $\phi$ then
\emph{fits} such a labeled example if
$M,w\models\phi$ and $\lab=1$, or if
$M,w\not\models\phi$ and $\lab=0$.
A set $E$ of such labeled examples 
\emph{uniquely characterizes} a formula $\phi$ 
relative to a fragment $\ML_\Phi[\PROP]$ and a logic $\Lambda$, if $\phi$ fits $E$ and every formula 
$\psi\in \ML_\Phi[\PROP]$ that fits $E$ is $\Lambda$-equivalent to $\phi$.
\begin{restatable}{theorem}
{thmModalchardichotomy}\label{thm:Modal_char_dichotomy}
    Fix $\Lambda=\logic{K}$.
    Let $\Phi$ be a simple set of modal formulas containing at least one of $\Diamond,\Box$. Then the following are equivalent:
    \begin{enumerate}
        \item  
        Every $\phi\in \ML_{\Phi} [\PROP]$  (with $\PROP$ a non-empty finite set of propositional variables) is uniquely characterized with respect to $\ML_{\Phi}[\PROP]$ by a finite set of labeled examples.
        \item 
        The concept class $\ML_{\Phi} [\PROP]$ (with $\PROP$ a non-empty  finite set of propositional variables) is effectively exactly learnable with membership queries.
        \item
        $\Phi \preceq \{\land,\lor,\top,\perp,\Diamond\}$ 
        or 
        $\Phi \preceq \{\land,\lor,\top,\perp,\Box\}$  or 
        $\Phi \preceq \{\land,\lor,\Diamond,\Box\}$ 
        or 
        $\Phi\preceq \{\neg,\perp,\Diamond,\Box\}$
        or 
        $\Phi\preceq \{\land, \top,\Diamond, \Box\}$
        or 
        $\Phi\preceq \{\lor, \bot,\Diamond, \Box\}$. %
    \end{enumerate}
\end{restatable}

The proof is given in Appendix~\ref{AppendixII}.

A natural refinement of the question addressed by Theorem~\ref{thm:Modal_char_dichotomy} is: which simple modal fragments admit \emph{polynomial sized} unique characterizations? Another related (not necessarily equivalent) open question is: which simple modal fragments are  efficiently exactly learnable with membership queries.

\section{Other logics}\label{sec:other_logics}
In this section, we highlight some results concerning expressivity for non-classical and non-Modal Logics. We state results from Dependence Logic \cite{vaananen2007dependence}, Hybrid Logic \cite{ARECES_TENCATE}, Finitely Valued Logics \cite{malinowski}, Intuitionistic Logic \cite{Nick_DeJong} and Linear Temporal Logic \cite{demri2016temporal}. We provide concise introduction is provided in case, however  readers are encouraged read the above cited literature for a thorough understanding.

\subsection{Finitely valued logics}
\label{sec:many-valued}

A $k$-valued logic is a logic whose expressions denote truth values taken from a set $A=\{a_1, \ldots, a_k\}$ of truth values. Thus, every connective in such a logic, as well as every formula in such a logic, can be viewed as denoting a finitary operation over $A$. Therefore, in order to study the relative expressive power of different (syntactic fragments of) $k$-valued logics, one must consider clones over a set
$A=\{a_1, \ldots, a_k\}$. These may also be called \emph{$k$-valued clones}. Note that Boolean clones are a special case of $k$-valued clones, namely where $k=2$.

Unfortunately, for $|A|>2$, the lattice $\mathcal{C}_A$ of all clones over $A$
is substantially more rich than the Post lattice. In particular, it
has cardinality of the continuum \cite{n_many},  and a complete representation similar to the one depicted in Figure~\ref{fig:Post_Lat} is not known (although partial results are known for some restricted sublattices, such as the lattice of clones of self-dual functions over a 3-element set~\cite{Zhuk}). 

On the bright side, the clone membership
problem and the clone containment problem are still decidable for $k$-valued clones.

\begin{fact}
    The following are decidable:
    \begin{enumerate}
    \item given a finite set $O\subseteq \finops{\{a_1,\ldots, a_k\}}$ and $f\in \finops{\{a_1,\ldots, a_k\}}$,  does $f\in \cloneof{O}$?
    \item given finite sets $O,O'\subseteq \finops{\{a_1,\ldots, a_k\}}$, is it the case that $O\preceq O'$?
    \end{enumerate}
\end{fact}

Note that the above two problems are computationally equivalent to each other.
The decidability argument is the same as for tabular modal logics: in order to test whether an $n$-ary operation $f$ belongs to $\cloneof{O}$, it suffices 
to generate all functions of arity at most $n$ belonging to $\cloneof{O}$, which can be done by starting with the projections and closing under the operations in $O$ through a recursive saturation process that terminates when no new functions of arity at most $n$ are generated. 
Indeed, the clone membership function for arbitrary finite domains is \class{EXPTIME}-complete if the input functions are given as truth tables \cite{Bergmann1999}, and 
\class{2EXPTIME}-complete if they are given succinctly as circuits \cite{Vollmer_Complexity}.

Criteria for expressive completeness were extensively studied in~\cite{kuznetsov1956,yablonski1958,maltsev1965,rosenberg1965,yablonski1986}.
In particular, it was Kusnetsov~\cite{kuznetsov1956} who first established a decideable criterion for
the expressive completeness of $k$-valued logics. 

\subsection{Intuitionistic Logic}
Intuitionistic Propositional Logic (\logic{IPC}) is, from a proof theoretic point of views, 
a sublogic of classical propositional logic obtained by dropping the Law of Excluded Middle (i.e., $\vdash p\lor \neg p$) \cite{CZ97}. It captures intuitionistic reasoning, where truth is understood in terms of constructive provability rather than classical bivalence. The algebraic semantics of $\logic{IPC}$ is given by Heyting algebras \cite{CZ97, Esakia19}. Thus, in the same way that we defined modal clones, for a normal modal logic $\Lambda$, as 
subclones of the clone generated by the 
Lindenbaum-Tarski algebra $\mathcal{A}_\Lambda$,  clones for \logic{IPC} can be defined as subclones of the clone generated by the free Heyting algebra (over a countably infinite set of generators). In this way, we again obtain
a correspondence between clones for \logic{IPC} and fragments $\logic{IPC}_{\Phi}$ (for $\Phi$ a set of $\logic{IPC}$-formulas). 

\begin{example}
    $\logic{IPC}_{\{\to\}}$ is the syntactic fragment of $\logic{IPC}$ consisting of formulas built up from propositional variables using only the connective $\to$, which is also known as the \emph{implicative fragment}.
    Similarly, $\logic{IPC}_{\{\to,\land\}}$ is known as the \emph{implicative-conjunctive fragment} and $\logic{IPC}_{\{\to,\land, \bot\}}$ as the \emph{implicative-conjunctive-negation fragment}.
\end{example}

Various results have been obtained for such fragments. For example, the extensions of $\logic{IPC}_{\{\to,\land\}}$ and $\logic{IPC}_{\{\to,\land, \bot\}}$ that have the Craig interpolation property have recently been fully characterized in \cite{Dekker2020Interpolation}. 

All the same questions we considered for modal fragments can now be considered for fragments of \logic{IPC} in the same way.  
Although relatively little is known, one 
important positive result was obtained in~\cite{Ratsa71} namely that
the \emph{expressive completeness} problem 
(i.e., given a finite set of formulas $\Phi$, 
is it the case that \emph{every} formula is provably equivalent in $\logic{IPC}$
to a formula in $\logic{IPC}_{\Phi}$)
is decidable. This was proved by showing that there are, up to equivalence, only finitely many expressively pre-complete fragments of the form $\logic{IPC}_{\Phi}$. It is interesting
to contrast this with Theorem~\ref{thm:GL-precomplete}, according to which there
are infinitely many expressively pre-complete modal fragments with respect to \logic{GL}.

\subsection{(Linear) Temporal Logic}
The formulas of Linear Temporal Logic (LTL) is built using a finite set of propositional variables $AP$, and closed under the propositional operators $\land,\neg$, and temporal modal operators $U,S, F, G$ and $X$. Formally, 
$$\varphi \Coloneqq p\;|\; \varphi \land \psi\;|\;\neg \varphi\;|\; \varphi U\psi \;|\; \varphi S \psi\;|\; X\varphi$$ 
The satisfaction of LTL formulas is through infinite sequences, having entries from $2^{AP}$. The satisfiability relation, $\models$ say, is defined as per convention for the propositional variables and operators (see \cite{Blackburn} Ch 6.3 for details). The modal operators are defined as follows; Suppose $w=w_0,w_1,w_2\dots$ is a (infinite) sequence, and denote $w^i=w_i,w_{i+1},w_{i+2}$, i.e. the suffix of sequence $w$, starting from $w_i$. 
\begin{align*}
    &w\models X\varphi, \text{ if }\; w^1\models \varphi\\
    &w\models \varphi U\psi ,\text{ if }\;\exists i (w^i\models \psi\text{ and for any } 0\leq k<i\; (w^k\models \varphi))\\
    &w\models \varphi S\psi, \text{ if }\; \exists k,v (v^k=w, v\models \varphi \text{ and for any $0\leq i <k$ }(v^i\models \psi))\\
    &w\models F\varphi \text{ if }\exists i (w^i\models \varphi)\\
    &w\models G\varphi \text{ if }\forall i (w^i\models \varphi)
\end{align*}
Temporal operators are in-fact modal operators, since the infinite sequences correspond to a particular class of Kripke frames. Moreover the temporal operator $U$ cannot be defined in basic modal Logic (see \cite{Blackburn}), proving LTL to be a richer modal language. 

The satisfiability problem of this richer modal language is known to be $\class{PSPACE}$-complete. It was additionally proven in \cite{Clarke_Sistla} that the problem is either $\class{NP}$-complete or $\class{PSPACE}$-complete, depending on the temporal operators allowed in the formula. Closely related to the satisfiability problem, is the model-checking problem.  
\begin{theorem}[\cite{Bauland2008:LTL}]
Let $M\subseteq\{F,G,X,U,S\}$ and let $O$ be a set of Boolean functions.
\begin{enumerate}
    \item If $O \preceq \{\lor,\leftrightarrow\}$ or $O \preceq \{\neg,\maj\}$, the satisfiability problem for LTL$_{M,O}$ is trivial.
    \item If $O \preceq \{\land,\lor,\top,\bot\}$ or $O \preceq \{\neg,\top,\bot\}$, the satisfiability problem for LTL$_{M,O}$ is solvable in polynomial time.
    \item If $\{x\land\neg y\}\preceq O$, then the satisfiability problem for LTL$_{M,O}$ is $\class{NP}$-complete if $M\subseteq\{F,G\}$ or $M=\{X\}$ and $\class{PSPACE}$-complete otherwise.
    \item Otherwise, $\{\oplus\}\preceq O\preceq \{\oplus,\bot\}$, and in this case the complexity of the satisfiability problem for $LTL_{M,O}$ is unknown.
    \end{enumerate}
\end{theorem}

A classification of the model-checking problem for LTL fragments has also been proposed in \cite{DBLP:journals/tocl/BaulandM0SSV11}. 
Furthermore, the temporal logics CTL and CTL$^*$ have been studied from this perspective~ \cite{DBLP:journals/ijfcs/MeierTVM09, DBLP:journals/ijfcs/MeierTVM15, DBLP:journals/acta/KrebsMM19}. 
Finally, a parameterized study of temporal logic fragments shows a largely negative picture concerning the parameters formula treewidth and pathwidth. Most fragments remain hard for a parameterized class called $\class{W}[1]$, and only fragments involving the $X$-operator are shown to be fixed-parameter tractable.

\subsection{Hybrid Logic}

Hybrid logics \cite{ARECES_TENCATE} take as their starting point the observation that the basic modal language (interpreted over Kripke structures) can  be viewed as a fragment of first-order logic. Specifically, hybrid languages are extensions of the basic modal language with operators that increase the expressive power of the language to bring it closer to the full expressive power of first order logic. 
Two operators that feature prominently in hybrid logics are the \emph{@-operator} (also known as the \emph{satisfaction operator}) and the \emph{$\downarrow$-binder}.
Together they allow us to write formulas  such as $\downarrow\! x.\Diamond\Diamond x$,
which holds at a world $w$ in a Kripke model if and only if $w$ is reachable from itself in two steps along the accessibility relation. 
The satisfiability problem of hybrid logics with the $\downarrow$-binder is known to be undecidable. In \cite{Meier2010:hybrid}, the authors investigate the effect of restricting the allowed Boolean connectives, as well as the allowed modal and hybrid operators, on decidability and  the complexity of the satisfiability problem over arbitrary, transitive, total frames, and frames based on equivalence relations (i.e., corresponding to the modal logics \logic{K}, \logic{K4}, \logic{T} and \logic{S5}). We can think of this as 
a hybrid-logic analogue of classification of simple modal fragments with respect to satisfiability provided by Theorem~\ref{thm:gensatmod}. 
Finally, in a successor paper~\cite{DBLP:conf/aiml/GollerMM0TW12}, the authors show that the satisfiability problem remains non-elementary over linear orders, but its complexity
drops to $\class{PSPACE}$-completeness over $\mathbb N$. 
They categorize the strict fragments arising from different combinations of modal and hybrid operators into NP-complete and tractable (i.e., complete for $\class{NC}^1$
or $\class{LOGSPACE}$).

\subsection{Nonmonotonic logics} 
Classical reasoning is monotonic, meaning that adding further information allows us to deduce the same or even more conclusions than before. 
Common-sense reasoning, however, is often nonmonotonic: adding new information may invalidate previous conclusions. Nonmonotonic logics aim to capture this type of reasoning. 
Examples of nonmonotonic logics include Default Logic \cite{DBLP:journals/ai/Reiter80}, Autoepistemic Logic \cite{DBLP:journals/ai/Moore93}, Circumscription \cite{DBLP:journals/ai/McCarthy80}, and various forms of Abductive Reasoning \cite{DBLP:reference/ml/Kakas10}. 
It is beyond the scope of this article to go into deep detail here, yet, we want to provide pointers to the literature where classifications are given under the scope of Post's lattice in this setting. 
Namely for Default Logic~\cite{DBLP:journals/logcom/BeyersdorffMTV12}, Autoepistemic Logic~\cite{DBLP:journals/tocl/CreignouMVT12}, Circumscription~\cite{DBLP:journals/mst/Thomas12}, Abduction~\cite{DBLP:journals/logcom/CreignouST12}. 
Recently, even a parameterized perspective in the setting of abduction has been analyzed~\cite{DBLP:journals/logcom/MahmoodMS21}.
Furthermore, logic-based argumentation~\cite{DBLP:conf/aaai/BesnardH00}, a formalism used in the modern areas of AI, has been studied from this perspective, recently~\cite{DBLP:journals/tocl/0002M023}.

\subsection{Modal Dependence Logic}
Modal dependence logic was recently introduced in \cite{Vaananen2008} by extending
modal logic by a dependence atom $\mathrm{dep}(\cdot)$. 
Semantic-wise, modal dependence logic follows the approach of classical (FO) dependence logic~\cite{vaananen2007dependence}. 
Here, the core idea is to evaluate formulas not with respect to single worlds in a Kripke model, but rather with respect to sets of worlds, called \emph{teams}. 
From this perspective, the dependence atom $\mathrm{dep}(p_1,\ldots,p_n,q)$ states that the truth value of proposition $q$ is functionally determined by the truth values of propositions $p_1,\ldots,p_n$ in the given team. 

In \cite{Muller13:model}, the authors investigate extended versions of \emph{modal
dependence logic}, parametrized by an arbitrary set of Boolean connectives.  Specifically, they study the computational complexity of the model checking problem, obtaining a complete classification that applies to all fragments of the language allowing for the modalities $\Diamond$ and $\Box$ and the dependence atom and an arbitrary set of Boolean functions as connectives.

\section{Conclusion}
\label{sec:conclusion}

In this paper, we surveyed research that systematically investigates syntactic fragments of propositional logic
and modal logics generated by a given set of operators. Our discussion of the propositional case is included
primarily as background: it provides the conceptual and technical template that motivates, and helps to frame,
the corresponding questions for modal fragments.

On the propositional side, the picture is well established. In particular, Post's lattice provides a crisp
organizing principle: expressive power is governed by the generated Boolean clone, and a wide range of
meta-questions and algorithmic tasks (e.g., satisfiability, tautology, counting, implication/equivalence,
evaluation, minimization) admit clean complexity classifications in terms of the position of the basis in
Post's lattice.

On the modal side, the most general analogous framework, allowing a basis of ``connectives'' defined by arbitrary
modal formulas, is algebraically natural (via Lindenbaum--Tarski algebras and modal clones), but quickly
becomes unruly: many of the attractive properties of Post's lattice fail for lattices of modal clones. In
particular, containment and expressibility problems are typically undecidable outside the locally tabular
setting. \emph{Simple} modal fragments---those obtained by choosing Boolean-function connectives together with a
specified subset of modal operators---form a well-behaved special case. They retain enough structure to
be able to take advantage of the organizing power of Post's lattice, leading to decidability and systematic classification
results for many standard problems, while still being expressive enough to cover a wide range of fragments
studied in modal and description logics, temporal logics, and related systems.

We also added a learning-theoretic dimension, collecting and extending results on
teachability/learnability from examples for propositional fragments and (simple) modal fragments. 

We close by highlighting a few open problems that
we identified in the paper or that naturally suggest themselves based on the story told in the paper.

\begin{enumerate}
    \item \textbf{Modal clones beyond the transitive setting.}
  For transitive modal logics, there is a  good understanding of which logics admit an
  effective analysis of modal clones, and hence of the fragments generated by modal connectives (cf.~Theorem~\ref{thm:expressibility} and Theorem~\ref{thm:expressibility-again}, respectively). In contrast,
  for non-transitive logics the overall landscape is less clear. In particular, the decidability of the containment problem for modal clones with respect to $\mathbf{K}$
  remains open.

    \item \textbf{The ``affine'' mystery and recurring classification gaps.}
Some of the classification results surveyed in Sections~\ref{ssec:complexity-log-reasoning} and \ref{sec:other_logics} are incomplete, in the sense that they lacked small cases where a complete classification had eluded the authors. 
Interestingly, the gaps in question are typically situated at the affine clones, i.e., the cases where $\{\threeXor\}\preceq O \preceq\{\xor,\top,\bot\}$. 
There is a list of publications that have open cases for such clones, e.g., \cite{DBLP:journals/tocl/BaulandM0SSV11,Bauland2008:LTL,DBLP:journals/logcom/CreignouST12,DBLP:journals/argcom/Creignou0TW11,Meier2010:hybrid,DBLP:phd/de/Reith2002,DBLP:journals/mst/Thomas12}. 
As a first step in approaching such affine cases, a recent work~\cite{krebs2025symmetryyieldsnphardnessaffine} has answered one of the questions posed in \cite{Gen_Sat_Mod}. Nevertheless, the mystery surrounding the affine clones still needs to be solved.

    \item \textbf{Beyond simple modal fragments}
    Let us call a modal formula \emph{shallow} if every occurrence of a proposition letter is in the scope of at most one modal operator. 
    Let us call a modal formula \emph{uniform-depth} if every occurrence of a proposition letter is in the scope of an equal number of modal operators. Every simple set consists of shallow uniform-depth formulas, but the converse is not true. In particular, the third and fifth example in Example~\ref{ex:modal-fragments} are shallow and uniform-depth and not simple. Could some of the positive results regarding simple modal fragments be extended to modal fragments defined by sets of
    formulas that are shallow and uniform-depth?
In particular, is the containment problem decidable for shallow, uniform-depth sets?

\item \textbf{Bisimulations for simple modal fragments}
Yde Venema (personal communication) suggested the following question, which to our knowledge remains open.
For several specific modal and description-logic fragments, suitably tailored variants of
bisimulation have been introduced that capture their expressive power; see, e.g., \cite{KurtoninaDeRijke1997SimulatingWithoutNegation, KurtoninaDeRijke1999Expressiveness}.
Is there is a general, Boolean-basis-parameterized notion of
bisimulation (or, more generally, model-comparison game) that characterizes
expressive equivalence for all simple modal fragments?
A plausible route is to combine the modal back-and-forth conditions of bisimulations with algebraic invariants
governing Boolean expressibility (i.e., polymorphisms
\cite{Jeavons1998AlgebraicStructure, JeavonsCohenGyssens1997Closure}). The categorical framework of
game comonads~\cite{AbramskyReggio2024:invitation, abramsky2025existentialpositivegamescomonadic}  may
 provide a suitable scaffolding for developing such a ``basis-aware'' bisimulation game.
\end{enumerate}

\begin{acks}
The authors would like to thank 
I think we should thank Andrei Russu, Alexei Muravitsky, Alex Citkin for helpful conversations
regrading results of Alexander Kuznetsov and Metodie Ra\c{t}\v{a};
Victor Dalmau and Raoul Koudijs for their feedback
on the results pertaining to learnability and teachability (which were obtained as part of Arunavo Ganguly's MSc thesis) and Heribert Volmer for helpful conversations. 
\end{acks}

\bibliographystyle{alpha}
\bibliography{bib}

\appendix
\section{Proof of Theorem~\ref{thm:prop-unique}}\label{AppendixI}

The positive direction of Theorem~\ref{thm:prop-unique} (i.e., the direction from 1 to 2)
will be given by the following
two results:

\begin{theorem}[\cite{Anthony92}]\label{thm:positive_analysis_monotone} For any $\PL_{\{\land, \lor, \top,\perp\}}$-formula $\phi$, the following are equivalent:
\begin{enumerate}
    \item
    $\phi$ is uniquely characterized with respect to $\PL_{\{\land, \lor, \top,\perp\}}$ by $m$ positive and $n$ negative examples.
    \item
    $\phi$ is equivalent to a DNF with $m$ terms, and to a CNF with $n$ clauses.
\end{enumerate}   
\end{theorem}

\begin{theorem}
\label{thm:positive_analysis_neg}
    Every $\PL_{\{\neg, \top\}}$ formula $\phi$ has a unique characterization of size 2. 
\end{theorem}
\begin{proof}
    Let $V_p$ be the assignment such that $V_p(x)=1 \iff x=p$. Similarly, let $V_\top$ be the assignment such that for every $x$, $V_\top(x)=1$, and $V_\perp$ be the assignment such that for every $x$, $V_\perp(x)=0$.
    
    It follows $\{(V_p,1), (V_\perp, 0) \}$ uniquely characterizes $p$, whereas $\{(V_p,0), (V_\perp,1)\}$ uniquely characterizes $\neg p$. $\top$ is   uniquely characterized by $\{(V_\top, 1), (V_\perp, 1)\}$, and $\perp$ is  uniquely characterized by $\{(V_\top, 0), (V_\perp, 0)\}$.
\end{proof}

The negative direction will be 
proved by establishing suitable lower bounds.

\begin{theorem}
\label{thm:negative_analysis_xor}
For every finite set $\PROP$ and for each $\PL_{\{\threeXor\}}[\PROP]$-formula $\phi$, every unique characterization of $\phi$ w.r.t. $\PL_{\{\threeXor\}}[\PROP]$ must contain at least $|\PROP|-1$ many examples.  
\end{theorem}

\begin{proof}
    Let $\phi$ be any $\PL_{\threeXor}[\PROP]$-formula, and 
    suppose for contradiction that there exists a uniquely characterizing set
$S$ for $\phi$ relative to $\PL_{\threeXor}[\PROP]$ with $|S|< n-1$ for $n=|\PROP|$. As a first step, we will upgrade $S$ to a set $S'$ that uniquely characterizes $\phi$ relative to the, slightly larger, propositional fragment 
$\PL_{\xor}[\PROP]$.
Let $S' = S \cup \{(v_\top,1)\}$, where $V_\top$ denotes the valuation that assigns $1$ to every
variable. It is easy to see that
every $\PL_{\threeXor}[\PROP]$-formula is true in $V_\top$, and hence $\phi$ fits $S'$.
We claim that $S'$ uniquely characterizes $\phi$ relative to
$\PL_{\xor}[\PROP]$.
Indeed, every $\PL_{\xor}[\PROP]$-formula computes a parity of some subset
$T\subseteq\{1,\dots,n\}$, i.e.\ the function
$x\mapsto \bigoplus_{i\in T} x_i$.  Evaluated at $V_\top$, this equals
$|T|\bmod 2$, so consistency with the positive example $(V_\top,1)$ forces
$|T|$ to be odd (an \emph{odd parity}). Every odd parity function has a 3XOR representation by grouping
its variables into triples and iterating the operation. Thus, within
$\PL_{\xor}[\PROP]$, the example $(V_\top,1)$ restricts attention exactly to
the functions definable in $\PL_{\threeXor}[\PROP]$.
 Since $S$ uniquely characterizes $\phi$
within $\PL_{\threeXor}[\PROP]$, it follows that $S'$ uniquely characterizes $\phi$ within
$\PL_{\xor}[\PROP]$.
Note that 
$|S'| < n$.

Next, we use the fact that 
$\PL_{\xor}[\PROP]$-formulas
have a natural linear-algebra
interpretation. We can 
view a valuation as a vector $x\in\mathbb{F}_2^n$, where
$\mathbb{F}_2$ is the two-element field, via $x_i=v(p_i)$.
By a \emph{parity function} we will mean a function $f_a$, for $a\in\mathbb{F}_2^n$, given by
\[
f_a(x) \;=\; a\cdot x \pmod 2,
\]
where $a\cdot x$ is the dot product over $\mathbb{F}_2$. 
We can now naturally view $S'=\{(V_1,lab_1), \ldots, (V_m,lab_m)\}$ as system of linear equations. 
More precisely, let $X$ be the $m\times n$-matrix where
$X_{i,j}=V_i(p_j)$, and let
$b$ be the length-$n$ vector consisting of the labels $lab_1, \ldots, lab_m$. Then 
the fitting $\PL_{\xor}[\PROP]$-formulas for $S'$ precisely 
correspond to the parity functions $f_a$ satisfying 
\[Xa=b \qquad \text{over $\mathbb{F}_2$}\]

Because $m<n$, the rank of the matrix $X$ is strictly less than $n$. 
It is a well-known fact in linear algebra that if a linear system $Xa=b$ 
is consistent and has rank $< n$, then solutions are not unique. Therefore, 
there must be multiple fitting $\PL_{\xor}[\PROP]$-formulas.
\end{proof}

We need similar lower bounds 
for several other propositional
fragments. To obtain these, 
we will use a suitable notion
of reduction that allows us 
to carry over these lower bounds.

\begin{definition}[PC reduction]
Let $\mathcal{C}_1=(C_1,E_1,\lambda_1)$ and $\mathcal{C}_2=(C_2,E_2,\lambda_2)$ be two concept classes. 
A \emph{reduction} from $\mathcal{C}_1$ to $\mathcal{C}_2$ is a pair of functions $(f,h)$ 
where $f\colon C_1\to C_2$ and $h\colon E_1\to E_2$ such that
\begin{enumerate}
    \item 
    For all $c\in C_1$ and $e\in E_1$, $e\in \lambda_1(c)$ iff 
    $h(e)\in \lambda_2(f(c))$, and
    \item 
    For each $e\in E_2$, either 
    $e$ belongs to $\lambda_2(f(c))$ for all $c\in C_1$, or $e$ belongs to  $\lambda_2(f(c))$ for no $c\in C_1$,
    or else there is an $e'\in E_1$ such that  $\{c\in C_1\mid e\in \lambda_2(f(c))\} = \{c\in C_1\mid e'\in \lambda_1(c))\}$.
\end{enumerate}
    We write $\mathcal{C}_1\leq_{pc} \mathcal{C}_2$ when such a reduction exists.
\end{definition}
Intuitively, the first item tells us that the pair of functions $(f,h)$ 
embed $\mathcal{C}_1$ into $\mathcal{C}_2$, while the second
item tells us that, when it comes to
distinguishing $\mathcal{C}_1$-concepts from each other, 
$\mathcal{C}_2$-examples offer no advantage over $\mathcal{C}_1$-examples.
The following propositions follow immediately
from the definition:

\begin{proposition}
   \label{prop:PC1}
    Let $\mathcal{C}_1=(C_1,E_1,\lambda_1)$ and $\mathcal{C}_2=(C_2,E_2,\lambda_2)$ be two concept classes such that $\mathcal{C}_1\leq \mathcal{C}_2$. If, for some concept $c\in C_1$, every uniquely characterizing set of labeled examples in $\mathcal{C}_1$ has size at least $n$, then the same holds for the concept $f(c)$ in $\mathcal{C}_2$.
\end{proposition}

\begin{proposition}
  \label{prop:PC2}
    $\PL_O[\PROP]\preceq \PL_{O'}[\PROP]$ holds whenever $O\preceq O'$.
\end{proposition}

We will now put these reductions to use.

\begin{proposition}
\label{prop:PC3}
    The following holds: 
    \begin{enumerate}
        \item $\PL_{\{\threeXor\}}[\PROP]\leq_{pc} \PL_{\{\oxor\}}[\PROP\cup \{w\}].$
        \item $\PL_{\{\threeXor\}}[\PROP]\leq_{pc} \PL_{\{\aimp\}}[\PROP\cup 
        \{w\}]$.
    \end{enumerate}
\end{proposition}

\begin{proof}
(1)
Every $\PL_{\{\oplus\}}$ formula can 
be written as $p_{i_1}\oplus \dots \oplus p_{i_n}$ (where $n$ is an odd number). Let
\[\begin{array}{ll}
f(\varphi) &=
\begin{cases}
 w\lor(w \xor p) &\text{if $\varphi=p$, where $p$ is a propositional variable.}\\
 w \lor (q\oplus \alpha_1(\theta))&\text{if $\varphi$ is equivalent to $\theta \oplus q$,  which contains no repetitions.}
\end{cases}\\
h(V) &=V\cup\{(w,0)\}
\end{array}\]
It is clear that $f(\phi)$ is equivalent to $\phi$ over valuations where $w$ is false, while $f(\phi)$
is always true over valuations where $w$ 
is satisfied. It follows that the two conditions of PC-reductions are satisfied.

(2) Since 
$\{\threeXor\}\preceq \{\land,\to\}$,
it suffices to give a PC-reduction
from $\PL_{\land,\to}[\PROP]$ to 
$\PL_{\aimp}[\PROP\cup\{w\}]$. Let
$f(\varphi) = w\land (w\rightarrow f'(\varphi))$ where
\[f'(\varphi)=
\begin{cases}
 p &\text{if $\varphi=p$, where $p$ is a propositional variable.}\\
 f'(\theta)\land(w\rightarrow  f'(\eta))&\text{if $\varphi=\theta \land \eta$.}\\
 w\land (f'(\theta)\rightarrow f'(\eta)) &\text{if $\varphi=\theta \rightarrow \eta$.}
\end{cases}\]
Furthermore, let 
$h(V)=V\cup \{\{w,1\}\}$. It is clear that $f(\phi)$ is equivalent to $\phi$ over valuations where $w$ is true, while $f(\phi)$ is always false
over valuations where $w$ is false. 
It follows that the two conditions of PC-reductions are satisfied.
\end{proof}

\thmpropunique*

\begin{proof}
The implication from 1 to 2 is given by Theorem~\ref{thm:positive_analysis_monotone} and Theorem~\ref{thm:positive_analysis_modal_neg} together
with Proposition~\ref{prop:PC1} and Proposition~\ref{prop:PC2}.
The implication from 2 to 3 is immediate. 
The implication from 3 to 1 is proved by contraposition. A visual inspection of the Post's Lattice reveals that, when $O\not\preceq \{\land,\lor,\top,\perp\}$ and $O\not\preceq \{\neg, \top\}$, then
either $\{\threeXor\}\preceq O$ or $\{\aimp\}\preceq O$ or $\{\oxor\}\preceq O$.
The first case is handled by 
Theorem~\ref{thm:negative_analysis_xor} while the second and third case follow by Proposition~\ref{prop:PC1} and Proposition~\ref{prop:PC3}.
\end{proof}

\section{Proof of Theorem \ref{thm:Modal_char_dichotomy}}
\label{AppendixII}

The proof of Theorem~\ref{thm:Modal_char_dichotomy}  builds on several prior results.

\begin{theorem}[\cite{tenCate-Koudijs}]
\label{thm:Modal-monotone}
Let $\PROP$ be any non-empty finite set of propositional variables. Then following hold:
\begin{enumerate}
    \item  $\ML_{\{\Diamond,\Box,\land, \lor\}}[\PROP]$,
    $\ML_{\{\Box,\land, \lor,\top,\perp\}}[\PROP]$ and $\ML_{\{\Diamond,\land, \lor,\top,\perp\}}[\PROP]$  admit finite characterizations.   Furthermore, the finite characterizations are effectively computable.
    \item Neither  $\ML_{\{\Diamond,\Box,\land ,\perp\}}[\PROP]$ nor $\ML_{\{\Diamond,\Box, \lor, \top\}}[\PROP]$ admits finite characterization.
 
\end{enumerate}
\end{theorem}

\begin{theorem}[\cite{tencate2024:power}]
\label{thm:modal-characterization-two-more-positive}
    Let $\PROP$ be a finite set of propositional variables. Then 
    $\ML_{\{\Diamond,\Box,\land, \top\}}[\PROP]$ and 
    $\ML_{\{\Diamond, \Box,\lor, \bot\}}[\PROP]$ admit finite characterizations. 
    Furthermore, the finite characterizations are effectively computable.%
    \footnote{In fact, the cited paper only states the result for $\ML_{\{\Diamond,\Box,\land, \top\}}[\PROP]$ but the other result follows by a straightforward dualization argument, as described in the proof of \cite[Theorem 3.9]{tenCate-Koudijs}}
    \end{theorem}
    \begin{theorem}\label{thm:positive_analysis_modal_neg} 
        Let $\PROP$ be any non-empty finite set of propositional variables.
        Then $\ML_{\{\Diamond, \Box,\neg, \top\}}[\PROP]$ admits finite characterization. Moreover, the finite characterizations are effectively computable.
    \end{theorem}

\begin{proof}
    Every formula in this fragment is of one of the following four forms, where $n\geq 0$ and each $\circ_i\in\{\Diamond,\Box\}$:
\begin{enumerate}
    \item $\circ_1 \ldots \circ_n p$ 
    \item $\circ_1 \ldots \circ_n \neg p$
    \item $\circ_1 \ldots \circ_n \top$ where either $n=0$ or $\circ_n=\Diamond$
    \item $\circ_1 \ldots \circ_n \bot$ where either $n=0$ or $\circ_n=\Box$
\end{enumerate}
    we will refer to $n$ as the \emph{depth} of the formula. 

    Formulas of type 1 with depth $n$ are distinguished by the fact that 
    they fit the positive example $(M_1,0)$ and the negative example $(M_2,0)$ where 
    $M_1$ is the (non-transitive) chain $\fbox{$0\to 1\to\ldots \to n^{~p}$}$ and $M_2$ is the (non-transitive) chain $\fbox{$0\to 1\to\ldots \to n$}$. Thus, such a formula $\phi$ can be uniquely characterized by taking these two labeled example plus finitely many additional examples  distinguish $\phi$ from other formulas of type 1 for the same value of $n$. 
    Formulas of type 2 admit an analogous analysis where the labels of the two examples are flipped.
    Formulas of type 3 with depth $n$ are distinguished by the fact that they fit the positive examples
    $(M_1,0)$, $(M_2,0)$, $(M_3,0)$ and $(M_4,0)$ where
    $M_1=\fbox{$0\rotatebox{60}{$\circlearrowleft$}$}$, 
    $M_2=\fbox{$0\rotatebox{60}{$\circlearrowleft$}^{\PROP}$}$, 
    $M_3=\fbox{$0\to 1\to \cdots n$}$ and $M_4=\fbox{$0\to 1\to \cdots n \rotatebox{60}{$\circlearrowleft$}\to n+1$}$
    Thus, again, such a formula $\phi$ can be uniquely characterized by taking these two labeled example plus finitely many additional examples  distinguish $\phi$ from other formulas of type 3 for the same value of $n$. Note that the positive example $(M_1,0)$ rules out all formulas of type 1 and that the positive example rules out all formulas of type 2 and 4, and 
    the positive examples $(M_3,0)$ and $(M_4,0)$ together rule out all 
    formulas of type 3 with depth strictly greater than $n$ where $\circ_n=\Diamond$ or $\circ_n=\Box$, respectively.
    Formulas of type 4 again admit an analogous analysis, using the same four examples as negative examples instead of as positive examples.
\end{proof}

The remaining ingredients of the proof of Theorem~\ref{thm:Modal_char_dichotomy} are several reductions and, as usual, a careful analysis of the Post lattice. Specifically,
we will again make use of PC-reductions introduced in Appendix \ref{AppendixI}.

\begin{lemma}
\label{lem:pc-reduction-modal}
    Let $O$ be any set of Boolean functions and $\PROP$ any finite set of propositional variables. The following hold:
    \begin{enumerate}
        \item  
        $\PL_O[\PROP]\leq_{pc} \ML_{O\cup\{\Diamond\}}[\{p\}]$
        \item
        $\PL_O[\PROP]\leq_{pc} \ML_{O\cup\{\Box\}}[\{p\}]$
    \end{enumerate}
    \end{lemma}

We remind the reader that when we write $\ML_{O\cup\{\Diamond\}}$ or
$\ML_{O\cup\{\Box\}}$, it is understood that the Boolean functions
in $O$ are replaced by  arbitrary
propositional formulas expressesing them).

\begin{proof}
    1. Let $\PROP=\{p_0\ldots, p_{n-1}\}$. The PC-reduction is given by $(f,g)$, where
    \[\begin{array}{ll}
    f(\phi) &= \phi[p_0\mapsto p, p_1\mapsto \Diamond p, p_2\mapsto\Diamond\Diamond p, \ldots, p_{n-1}\mapsto\Diamond^{n-1} p] \\
    g(V)    &= (\mathbb{L}_n(V),0)
    \end{array}\]
    where, for any propositional valuation $V$,
    $\mathbb{L}_n(V)$ denotes the (non-transitive) chain 
    $$
    \fbox{$0\to 1\to\ldots\to n-1$}\, ,
    $$ 
    where world $i$ satisfies $p$ if and only if $V(p_i)=1$.
    Clearly, $g(V)\models f(\phi)$ iff $V\models\phi$. Furthermore,
    for each pointed Kripke model $(M,w)$, let $V_{M,w}$ be the $\PROP$-valuation given by 
    $V_{M,w}(p_i)=1$ iff 
    $M,w\models\Diamond^i p$. Then
    it is clear that $M,w\models f(\phi)$ iff $V_{M,w}\models\phi$.
    Therefore, $(f,g)$ is a valid PC-reduction.

    2. The proof is analogous, using $\Box$ instead of $\Diamond$.
\end{proof}

In combination with the results 
about $\PL_{\threeXor}$, $\PL_{\oxor}$ and $\PL_{\aimp}$ 
from Appendix~\ref{AppendixI},
this yields:

\begin{proposition}
$\ML_{\threeXor,\Diamond}[\{p\}]$, 
$\ML_{\threeXor,\Box}[\{p\}]$, 
$\ML_{\oxor,\Diamond}[\{p\}]$, 
$\ML_{\oxor,\Box}[\{p\}]$, 
$\ML_{\aimp,\Diamond}[\{p\}]$, and $\ML_{\aimp,\Box}[\{p\}]$ do not admit finite characterizations. In fact, in each of these fragments, the atomic formula $p$ cannot be 
uniquely characterized by a finite set of examples.
\end{proposition}

\begin{proof}
The results in Appendix~\ref{AppendixI}
show that, in each of the  propositional fragments
$\PL_{\threeXor}[\PROP]$, 
$\PL_{\oxor}[\PROP]$, 
$\PL_{\aimp}[\PROP]$, the atomic formula
$p_0$ already requires $\Omega(|\PROP|)$ many labeled examples. It follows by Proposition~\ref{prop:PC1} and by the proof of Lemma~\ref{lem:pc-reduction-modal} (since the reduction in question always maps $p_0$ to $p$ regardless of the choice of $\PROP$)
that the modal formula
$p$ cannot be uniquely characterized by a finite set of examples in each of the six stated modal fragments.
\end{proof}

We are now ready to commence with the proof of Theorem~\ref{thm:Modal_char_dichotomy}.

\thmModalchardichotomy*

\begin{proof}[Proof] 
It is a well known fact that a concept class is effectively (but not necessarily efficiently) exactly learnable with membership queries if and only if the concept class admits finite characterizations, and there is an effective algorithm that computes, for each concept, a uniquely characterizing finite set of labeled examples. 
In one direction, an  effective exact learning algorithm with membership queries can be used to construct a finite characterization for any concept, namely consisting of the examples appearing in membership queries posed by the algorithm before identifying the concept in question. 
To be precise, this also assumes that we can effectively test whether a given example is positive or negative for a given concept, an assumption that is indeed satisfied in our case. 
Conversely, every effective algorithm for constructing finite unique characterizations yields a brute-force exact learning algorithm: we simply enumerate all hypotheses, and 
check them one by one by testing the examples in their uniquely characterizations using membership queries.
This establishes the direction from 2 to 1. In combination with Theorem~\ref{thm:Modal-monotone} and~Theorem~\ref{thm:positive_analysis_modal_neg} it also shows that 3 implies 2. 

It only remains to show that 1 implies 3. This implication will be proved by contraposition.
Careful inspection of the Post lattice shows that,
for every simple set of modal formulas $\Phi$,
 one of the following two cases holds:
\begin{itemize}
    \item 
    $\Phi\preceq \{\Diamond, \land, \lor, \top,\perp\}$ or $\Phi \preceq \{\Box, \land, \lor, \top,\perp\}$ or $\Phi\preceq \{\Diamond,\Box, \land, \lor \}$ or $\Phi \preceq \{\Diamond, \Box, \neg, \top\}$ 
    or 
    $\Phi \preceq \{\Diamond, \Box, \land, \top\}$
    or
    $\Phi \preceq \{\Diamond, \Box, \lor, \perp\}$.
    \item 
    $\{\Diamond, \aimp\}\preceq \Phi$ or $\{\Diamond, \oxor\}\preceq \Phi$ or $\{\Diamond, \threeXor\}\preceq \Phi$ or  $\{\Box, \oxor\}\preceq \Phi$ or $\{\Box, \aimp\}\preceq \Phi$ or $\{\Box, \threeXor\}\preceq \Phi$ or $\{\Diamond, \Box,\land, \perp \}\preceq \Phi$ or $\{\Diamond, \Box,\lor, \top\}\preceq \Phi$.    
\end{itemize}
Since the first case doesn't apply, we must be in the second case. 
We had already observed that 
$\ML_{\threeXor,\Diamond}[\PROP]$, 
$\ML_{\threeXor,\Box}[\PROP]$, 
$\ML_{\oxor,\Diamond}[\PROP]$, 
$\ML_{\oxor,\Box}[\PROP]$, 
$\ML_{\aimp,\Diamond}[\PROP]$, and $\ML_{\aimp,\Box}[\PROP]$ do not admit finite characterizations.
It is clear that $\ML_\Phi[\PROP]\preceq \ML_{\Phi'}[\PROP]$ holds whenever $\Phi\preceq \Phi'$, so this takes care of the first six possibilities.
Similarly, if $\{\Diamond, \Box,\land,\perp \}\preceq\Psi$ or $\{\Diamond, \Box, \lor, \top\}\preceq\Psi$,
then it follows from Theorem \ref{thm:Modal-monotone}  that $\ML_\Psi[\PROP]$ does not admit finite characterization. 
\end{proof}

\end{document}